\documentclass[12pt]{article}
\usepackage{amsmath,amsthm,amsfonts,amssymb,mathrsfs,epsfig,bm}
\usepackage{soul}
\usepackage{graphicx}
\usepackage{subfig}

\usepackage[usenames, dvipsnames]{color}
\usepackage[colorlinks=true]{hyperref}

\newtheorem{theorem}{Theorem}
\newtheorem{lemma}{Lemma}
\newtheorem{corollary}{Corollary}

\newtheorem{definition}{Definition}

\newtheorem*{theorem*}{Theorem}
\theoremstyle{remark}
\newtheorem{remark}{\bfseries{Remark}}

\def\Z{\mathbb{Z}}

\def\R{\mathbb{R}}

\def\P{\mathbb{P}}
\def\E{\mathbb{E}}
\def\FF{\mathscr{F}}

\def\BB{\mathscr{B}}

\def\SS{\mathscr{S}}

\renewcommand{\phi}{\varphi}
\renewcommand{\epsilon}{\varepsilon}
\newcommand{\1}{{\text{\Large $\mathfrak 1$}}}
\newcommand{\comp}{\raisebox{0.1ex}{\scriptsize $\circ$}}

\newcommand{\eqdist}{\stackrel{\text{\rm (d)}}{=}}

\definecolor{mygray}{gray}{0.4}
\definecolor{deeppink}{RGB}{255,20,147}
\definecolor{mygreen}{rgb}{0.0, 0.75, 0.0}
\definecolor{myred}{rgb}{0.768, 0.09, 0.09}

\newcommand{\red}{\color{red}}

\long\def\symbolfootnote[#1]#2{\begingroup
\def\thefootnote{\fnsymbol{footnote}}\footnote[#1]{#2}\endgroup}
\newcommand{\keywords}[1]{ \noindent {\footnotesize
             {\small \it  Keywords and phrases.} {\sc #1} } }
\newcommand{\ams}[2]{  \noindent {\footnotesize
             {\small \it  AMS {\rm 2010} subject classification.
             {\rm Primary {\sc #1}; secondary {\sc #2}} } } }
\newcounter{para}

\def\wE{{\widetilde \E}}
\def\wP{{\widetilde \P}}
\definecolor{eqcol}{RGB}{255,10,130}
\newcommand{\paintme}{\color{eqcol}}

\def\e{\mathbf e}

\def\Estar{{\displaystyle \E^*}}
\def\Pstar{{\displaystyle \P^*}}

\def\PO{{\mathcal P}}
\def\BL{{\mathcal B}}

\def\arr{{\mathfrak a}}

\def\LIFO{\text{LIFO}}
\def\pLIFO{\text{pLIFO}}
\def\FIFO{\text{FIFO}}

\begin{document}

\title{\Large \bf The distribution of age-of-information performance
measures for message processing systems} 
\author{\sc George Kesidis\thanks{\href{mailto:gik2@psu.edu}{gik2@psu.edu}} 
\and 
\sc Takis Konstantopoulos\thanks{\href{mailto:takiskonst@gmail.com}{takiskonst@gmail.com}} 
\and 
\sc Michael A.\ Zazanis\thanks{\href{mailto:zazanis@aueb.gr}{zazanis@aueb.gr}}}

\date{\small 22 January 2020}
\maketitle

\begin{abstract}
The idea behind the recently introduced ``age of information'' performance measure 
of a networked message processing system is that it indicates our knowledge
regarding the ``freshness" of  the most recent piece of information
that can be used as a criterion for real-time control.
In this foundational paper, we examine two such measures, one that has
been extensively studied in the recent literature and a new one
that could be more relevant from the point of view of the processor.
Considering these measures as stochastic processes in a stationary environment 
(defined by the arrival processes, message processing times and admission controls
in bufferless systems),
we characterize their distributions using the Palm inversion formula.
Under renewal assumptions we derive explicit solutions for their
Laplace transforms and show some interesting decomposition properties. 
Previous work has mostly focused on computation of expectations in very particular cases.
We argue that using bufferless or very small buffer systems is best and
support this by simulation. We also pose some open problems 
 including assessment of 
enqueueing policies that may be better in cases where one wishes to
minimize more general functionals of the age of information measures.
\\[2mm]
\keywords{Age of information; message processing systems; Palm probability; renewal process;
Poisson process; performance evaluation; stochastic decomposition}
\\[2mm]
\ams{60G55,60K05}{60G50,60K30}
\end{abstract}



{\small
\tableofcontents
}

\section{Introduction}
\label{sec:intro}
\subsection{Technological background}
The Internet is now commonly used to transmit latency-sensitive
information that is part of a real-time control or decision process.
As an example, consider a temperature or pressure sensor which could periodically
transmit a reading to a latency-critical
remote control. 
Other examples include decision systems for an airplane, driverless vehicles,
financial transactions, power systems, sensor/actuator systems 
or other ``cyber physical" systems.
In the power system case, a high temperature reading of a transmission
line could indicate reduced capacity or predict near-term failure.
In the sensor system case, the sensor could indicate an alarm 
such as a motion detector which needs to be manually reset once tripped; 
any alarm message would render stale any queued or 
in-transmission ``heartbeat" message 
that is periodically sent to indicate no intruder is present and that
the sensor is properly functioning.
In the actuator system case, messages may embody commands to a remote actuator
of a time-critical control system.

\subsection{Two age of information measures}
Systems such as the ones described above naturally depend on the {\it age} of the 
most recently received reading from a remote sensor. This is a quantity
that takes into account  the time since the reading was generated.
In view of the speeds involved 
a decision must be taken upon arrival of
a new information packet: to read or not read it. The choice is crucial and
depends on the packet length and the frequency of information packet arrivals,
quantities that may not be completely known. 
If 
$A^*_t:=$ {\bf\em arrival time of the last completely read message before time $t$}
then the quantity 
\[
\alpha(t):=t-A^*_t
\]
has been introduced in the literature and has been given the name 
{\bf\em ``age of information (AoI)''}.
This has been used as a measure
of freshness and its expectation (under specific assumptions) 
has been studied in, e.g.,
 \cite{Yates12,Eph16,Kosta17,Shroff17,Altman18}. 
From a performance point of view,
we are interested not only in its expectation but also in its probability distribution.
We derive fundamental results about the latter in this paper.

One can argue that the above measure may have limited usefulness for applications
that {\bf\em  cannot control the arrivals} 
of messages. And thus, one may assert that the
freshness of information should be gauged not against the current time $t$
but against
$A_t:=$ {\bf\em  arrival time of last message before $t$.}
By definition, $A^*_t \le A_t$ with equality if and only if the message
arriving at $A_t$ is completely read.
We thus introduce the measure 
\[
\beta(t):=A_t-A^*_t.
\]

To further explain our claim that $\beta$ may be more relevant than $\alpha$,
consider the following scenario: messages arrive rarely and randomly at times
$t_n$ and have 
very  small duration $\epsilon$, so small that $\epsilon \ll t_{n+1}-t_n$ for
all $n$.
Then only one message will be in the system at a time
and, assuming that the processor does not idle when the message is present,
every message is completely read. 
We can easily see that, unless $t$ is in the extremely small
interval of length $\epsilon$ during the processing of a message,
$\alpha(t)$ equals time elapsed between the last arrival before $t$ and $t$,
whereas $\beta(t)=0$.
(To see this, take  $t_0+\epsilon < t < t_1$. Then $A_t=A^*_t=t_0$,
so $\alpha(t)=t-t_0$, but $\beta(t)=0$. On the other hand,
if $t_1<t<t_1+\epsilon$
then $A_t=t_1$, but $A^*_t = t_0$. Hence
$\alpha(t)=t-t_0$, $\beta(t)=t_1-t_0$, and the two quantities
are approximately the same since $\epsilon$ is extremely small.)

Thus, $\alpha(t)$ and $\beta(t)$ are almost the same when $t$ lies in
a processing interval, but vastly different when $t$ lies in an idle interval.
In the latter case, $\alpha(t)$ simply tells us the age of the arrival process
but $\beta(t)=0$ meaning that the processor possesses the freshest information.
Thus, in situations where the arrival process is beyond the processor's control,
trying to keep the ``age of information'' low should not take into account the 
age of the arrival process. This is why we propose the new measure as a more 
relevant quantity. Granted, $\alpha(t) = (t-A_t)+\beta(t)$,
so, insofar as expectations are desired cost functionals, there is no difference
in potential optimization problems. However if the cost functional is another 
function $\alpha(t)$, e.g., $\P(\alpha(t)>u)$, then the dependence between
$t-A_t$ and $\beta(t)$ justify finding the distribution of $\beta$.
Since there is no terminology for this quantity, we are free to choose one:
we call it {\bf\em  ``new age of information (NAoI)''}.

\subsection{The queueing system: bufferless instead of buffered}
\label{bb}
The age of information measures can be defined for a general queueing system
that could consist of a number of queues and servers, buffers of various sizes
and various policies that control the acceptance of a message and its successful
processing. We define some quantities used in the paper.
Messages arrive at times $T_n$ and have processing (or service) times $\sigma_n$.
An arriving message may be immediately accepted or rejected. The $1/0$-valued
variable $\chi_n$ denotes acceptance/rejection. 
A message is called successful if it is processed in its entirety.
The $1/0$-valued variable $\psi_n$ denotes processing success/failure.
A failed message is kicked out of the system before it is read entirely.
We let $T_n'$ be the time at which the message arriving at time $T_n$ departs
from the system either because it is rejected or because it fails to be
processed entirely or because it departs successfully. 

Two systems that are of main concern in this policy are as follows.
There is a single server and a buffer of unit size (just to accommodate
the message being processed).
The first system operates under the {\bf\em  pushout policy ($\PO$)}.
Every arriving message immediately kicks out (one uses the word ``obsoletes'') the existing
message, if any, and starts being processed immediately. If no message
arrives while one is processed then the latter message finishes successfully.
Note that this system can be thought of as a Preemptive Last In First Out system
with buffer of size 1.
The second system operates under the {\bf\em  blocking policy ($\BL$)}.
An arriving message immediately grabs the server if the latter is available
or is immediately kicked out if the server is busy. Other policies are possible;
see the examples at the end of Section \ref{buff}.

The literature so far has focused on the AI for single server queueing  systems,
particularly stable FIFO or  preemptive LIFO disciplines \cite{Yates12,Shroff17},
with infinite buffers where all messages are accepted ($\chi_n=1$ for all $n$)
and all messages are successful ($\psi_n=1$ for all $n$).
Moreover, only the mean of $\alpha(t)$ for M/M/1-FIFO \cite{Yates12} in steady-state has been derived: see formula \eqref{mean-fifo} in the last section.
We do, however, question the use of infinite buffers, based on some simple,
intuitive observations.
The most intuitive of all is: if it is desired to keep the age of information low then
storing arriving message makes no sense as they contribute nothing to either $\alpha$
or $\beta$. 

Consider the $\PO$ system as described above and compare it with an
infinite buffer preemptive LIFO (pLIFO) system. Assume that the same
sequence of arrival and processing times is fed into both systems. 
Then, as explained in more detail in the last section,
\begin{equation}
\label{pLIFOvsP}
\alpha_\pLIFO(t) = \alpha_\PO(t), \quad 
\beta_\pLIFO(t) = \beta_\PO(t), \quad t \in \R.
\end{equation}
In fact, recently, it was shown that, among all work-conserving processing 
disciplines, for an infinite buffer single server queue, the preemptive LIFO discipline achieves stochastically lowest AoI in steady-state in some cases; see \cite{Shroff17}.
Concerning next an infinite buffer FIFO system, the other system
studied in the age of information literature \cite{Yates12,Mbook},
we conjecture that another system that we call $\PO_2$, basically a variation
of $\PO$ but with buffer size $2$, 
has better AoI performance than the infinite buffer FIFO.
It is for these reasons that we study  only bufferless systems in this paper.
In studying bufferless systems, the only variable is the queueing policy.
Rather than studying an optimal control problem,
we focus on two very specific and, in some sense, extreme policies,
$\PO$ and $\BL$.
We do so in order to obtain concrete formulas and explain the methods.
However, in principle, our methods, based on Palm calculus
and renewal theory, will work on any policy.


\subsection{Paper organization and contributions}
The paper is organized as follows.
In Section \ref{sec:mi} we present the setup and the definition of the models
and all relevant stochastic processes.
Section \ref{sec:outline} is a brief outline of some of the results, pointing out,
in particular, some interesting {distributional stochastic decomposition results}
for the various stationary performance measures.
Formulas for distributions and moments
of both the AoI and the NAoI for the pushout system
are derived in Section \ref{sec:fpo}.
This is done by carefully applying classical Palm theory, first in a stationary context
and then by specializing to the case involving independence assumptions.
The stronger the assumptions, the more explicit the results.
For the queueing theorist, it is not a surprise that the formulas become
quite explicit when the arrival process is Poisson.
Similarly pleasing and explicit is the case when the message
lengths are independent exponentially distributed random variables.
If both Poissonian assumptions hold then we are in the best of all worlds.
The blocking system is the subject of  Section \ref{sec:fbl}. 
The action plan is the same as in the pushout system case, but, here,
all calculations are more involved. This is due to the fact that
the blocking system has
more complicated dynamics than the pushout system.
Nevertheless, closed-form formulas are also possible.
In Section \ref{sec:fw}, we discuss variations of
the AoI problem to be considered in future work;  in particular,
we discuss other
queueing policies that may have smaller (in some sense) age of information
in some cases.

The contributions of this paper are as follows:
Previous literature has focused only on stationary mean AoI but
even that is done in rather specific cases (infinite buffer FIFO). 
In this paper we  {derive formulas for the distributions
via  Laplace transforms} of AoI and NAoI in steady-state
under renewal assumptions.
In particular, we find explicit formulas for all the means in all cases,
and even this appears to be novel. 
In addition to deriving formulas for the stationary distributions 
under renewal assumptions, by adopting
a top-down approach based on Palm calculus 
we derive a methodology on how one could
compute the same things (i) for arrival/service distributions with dependencies
and (ii) for policies other than $\PO$ or $\BL$.

\section{System definitions}
\label{sec:mi}
The goal of this section is to define the two measures of the age of information
for a general bufferless processing system.
We are careful to include the possibility that some of the quantities below
may be restricted on a lattice. 
We first define such a system, allowing the possibility to accept or reject messages.
We then give the definitions of the age of information measures as functions of
time.
Lastly, we introduce stochastic assumptions which make the age of information
processes random functions of time. 
\subsection{Notation/terminology}
The set of integers is denoted by $\Z$.
The indicator function of a set $A$ is denoted by $\1_A$.
The notation $\E[X; A]$ stands for $\E[X \1_A]$.
If $S$ is a set and $s \in S$, then $\delta_s$ denotes the Dirac measure
$\delta_s(B) = \1_{s \in B}$, $B \subset S$.
By point measure on $\R$ (or $\R^2$) we mean a measure assuming nonnegative
integer values; necessarily, it is a finite or countable sum of Dirac measures.
A point process is a random point measure.
If $X$ is a positive random variable with finite expectation, we say that
$\overline X$ is the stationary version of $X$ if it has density
$\P(X>x)/\E X$:
\[
\P(\overline X \in dx) = \frac{\P(X > x)} {\E X}  dx.
\]
We then have  \[
\E e^{-u \overline X} = \frac{1-\E e^{-u X}}{u \E X},\quad
\E \overline X = \frac{\E X^2}{2\E X}.
\] 
When $X$ and $Y$ are random variables (on, possibly, different probability spaces) 
$X \eqdist Y$ denotes equality of their laws (distributions).
The symbol $\wP$ denotes the probability governing a time-stationary system,
whereas $\P$ denotes the Palm probability  of $\wP$ with respect to the arrival process. See section \ref{stfr} below for exact definitions.
(We choose this unconventional notation because the former symbol
is used less frequently than the latter.)

\subsection{Bufferless message processing systems}
\label{buff}

Messages arrive in a bufferless server which can read one message at a time. 
Denote by $T_n$, $n \in \Z$, the message arrival times. We assume that
\[
T_n < T_{n+1}, ~ n \in \Z, \qquad
\sup_{n \in \Z} T_n =+\infty, \qquad \inf_{n \in \Z} T_n = -\infty.
\] 
We shall fix an ordering by letting $T_0$ be such that $T_0 \le 0 < T_1$.
We denote by
\[
\arr := \sum_{n \in \Z} \delta_{T_n}
\]
the arrival process, considered as a point measure.
We shall also let, for all $n \in \Z$,
\begin{equation}
\label{taun}
\tau_n := T_{n+1} - T_n.
\end{equation}
We introduce, for each $n \in \Z$, the {\it accept/reject index} $\chi_n$,
setting
\[
\chi_n = \begin{cases}
1, & \text{if the message arriving at $T_n$ is accepted}
\\
0, & \text{otherwise}.
\end{cases}
\]
The $\chi_n$ is a decision variable that depends on the
{\bf\em acceptance policy}. See below for some example.
In this paper we shall only consider specific policies leaving
optimization/control problems for future work.
The length  of message $n$ (the message arriving at time $T_n$) is
denoted by $\sigma_n$ and its departure time by $T_n'$. The latter given by
\begin{equation}
\label{T'}
T_n' :=
\begin{cases}
T_n, & \text{ if } \chi_n=0
\\
(T_n+\sigma_n) \wedge \inf\{T_r :\, r > n,\, \chi_r=1\},
& \text{ if } \chi_n=1
\end{cases}.
\end{equation}
This means that an arriving message will either be immediately rejected (and
thus depart immediately) or accepted, in which case it will either be read in
its entirety or pushed out by another accepted message.
Note that the sets $\{T_n, n \in \Z\}$ and $\{T_n', n \in \Z\}$
may have common elements (e.g., if we allow all variables take
values that are integer multiples of a common unit).
It is easy to see from \eqref{T'} 
that the intervals $[T_n, T_n')$ and $[T_m, T_m')$ are disjoint 
if $m \neq n$.
Thus, for all $t$, the quantity 
\begin{equation}
\label{qdef}
q(t) := \sum_{n \in \Z} \chi_n\, \1_{T_n \le t < T_n'}
\end{equation}
is either $0$ or $1$. The $q(t)$  is the state of the server at time $t$:
$q(t)=1$ if the server is busy or $0$ if not.
Notice that $q(\cdot)$ is right-continuous (by choice rather than by necessity).

We call message $n$ {\it successful} if it departs immediately after having being read
in its entirety. The {\it success/failure index} is the binary variable
\begin{equation}
\label{psi}
\psi_n := \1_{T_n' = T_n+\sigma_n}.
\end{equation}
By definition, for all $n$,
\[
\psi_n \le \chi_n.
\]
See Figure \ref{fig:T-fig} for an illustrative example of an arbitrary policy.
\begin{figure}
\centering
\includegraphics[width=10cm]{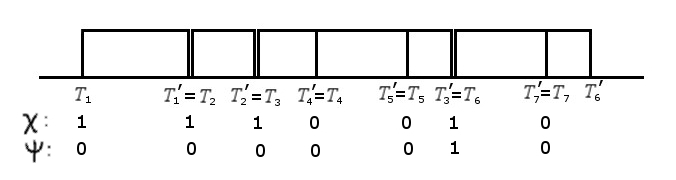}
\caption{\it A message arrives at time $T_1$ at an idle server and is immediately accepted.
A double line indicates that a message pushes out the previous one,
while a single line indicates that the message is blocked.
Thus, messages 1, 2, 3 and 6 are accepted, while 4, 5 and 7 are rejected.
Only message 6 is successful. The server started reading message 1 at time
$T_1$ and finishes reading message $6$ in its entirety at time $T_6'=T_6+\sigma_6$.}\label{fig:T-fig}
\end{figure}


Consider $n \in \Z$ and the statement
\begin{equation}
\mathcal Z_n := \text{``$q(T_n-)=0$ or  $T_m'=T_n$ 
for some $m < n$''}.
\label{Zn}
\end{equation}
We can interpret $\mathcal Z_n$ as ``the server is idle at time $T_n$''.
If there is no possibility that a departure time coincides with the
arrival time of another message then idle server simply means $q(T_n-)=0$.
But we must include the possibility that some message $m<n$ departs
exactly at $T_n$.
We shall throughout assume that the {\it non-idling condition} 
\begin{equation}
\tag{NI}\label{NI}
\text{for all } n \in \Z \text{ if } \mathcal Z_n \text{ then } \chi_n=1
\end{equation}
holds.
For those $n$ for which $\mathcal Z_n$ is violated the determination of 
$\chi_n$ is a matter of the acceptance policy.

Here are four examples of acceptance policies. Let $\ell$ be a nonnegative integer.

\paragraph{Example 1. The {\bf\em  pushout} ($\PO$) policy.}
All messages are accepted:
\[
\chi_n=1, \quad n \in \Z.
\]
From \eqref{T'} and \eqref{psi} 
it is easy to see that 
\[
\psi_n=\1_{T_n+\sigma_n \le T_{n+1}} = \1_{\tau_n \ge \sigma_n}, \quad n \in \Z.
\]

\paragraph{Example 2. The {\bf\em  blocking} ($\BL$) policy.} 
No message other than those satisfying the non-idling condition \eqref{NI}
are accepted:
\[
\chi_n =1 \iff \text{ $\mathcal Z_n$ holds}. 
\]
Note that, here, $\psi_n= \chi_n$ for all $n$, that is, every accepted message
is successful.

\paragraph{Example 3. The  $\BL\PO(\ell)$ policy.} 
Say a message arrives at time $t$ at an empty system, $q(t-)=0$.
Then it starts being processed. If there are at most $\ell$ arrivals 
while the message is being processed then they are all blocked.
Beyond that, the server accepts every arrival until it becomes empty again.
In other words, during a reading period, the server behaves in a blocking
fashion for up to $\ell$ arrivals and in a pushout fashion after that.

\paragraph{Example 4. The $\PO\BL(\ell)$ policy.} 
During a reading period, the server behaves in a pushout
fashion for up to $\ell$ arrivals and in a blocking fashion after that.

\paragraph{}We shall only study the first two policies in this paper,
leaving the study of the others, as well as optimal
policies, for future work.

\subsection{Age of information processes}
To define the age of information functions (of time) we 
need to introduce the following functions on $\R$.
The {\it last arrival epoch} before $t \in \R$ is defined by
\[
A_t := \sup\{T_n:\, n \in \Z, T_n \le t\}.
\]
The {\it last successful arrival epoch} before $t$ is defined by
\[
S_t := \sup\{T_n:\, n \in \Z, T_n \le t, \psi_n=1\};
\]
The {\it last successful departure epoch} before $t$ is defined by
\[
D_t := \sup\{T_n+\sigma_n:\,  n \in \Z, T_n+\sigma_n \le t, \psi_n=1\}.
\]
Note that, under our assumptions on the sequence $T_n$, 
the $\sup$ in the definition of $A_t$
is actually a $\max$. 
Assuming further that  
\[
\tag{A1}\label{A1}
\inf\{n: \psi_n=1\}=-\infty
\]
we have that the $\sup$
in $S_t$ and $D_t$ is replaced by a $\max$.
If, in addition, 
\[
\tag{A2}\label{A2}
\sup\{n: \psi_n=1\}=\infty
\]
then $S_t, D_t < \infty$ for all $t$.
\begin{definition}
\label{defaoi}
Under assumptions \eqref{A1} and \eqref{A2},
the {\bf\em  age of information (AoI)} function is defined by
\begin{equation}
\label{defalpha}
\alpha(t) := t -S_{D_t}, \quad t \in \R,
\end{equation}
and the {\bf\em  new age of information (NAoI) function} is defined by
\begin{equation}
\label{defbeta}
\beta(t):= A_t-S_{D_t}, \quad t \in \R.
\end{equation}
\end{definition}

Note that the functions $A, S, D$ above are right-continuous and increasing
($s < t \Rightarrow A_s \le A_t, S_s \le S_t, D_s \le D_t$). 
It follows that $\alpha$ and $\beta$ are also right-continuous.
Moreover,
\[
\Delta \alpha(t) := \alpha(t)-\alpha(t-) = - \Delta S_{D_t}
= - \lim_{\epsilon \downarrow 0} (S_{D_t}- S_{D_{t-\epsilon}}) \le 0.
\]
So jumps of $\alpha$ can only be negative.
Notice that
\[
\Delta \alpha(t) = S_{D_t} - S_{(D_{t-})-}.
\]
On the other hand, $\beta$ can have both positive and negative jumps.

We shall also use the following notations and terminology. 
Consider the arrival times $T_n$ of messages arriving at a idle server:
\[
\{B_k:\, k \in \Z\} := \{T_n:\, 
\text{ $\mathcal Z_n$ holds}\}.
\]
By convention, we enumerate these points as
\[
\cdots < B_{-1} < B_0 \le 0 < B_1 < \cdots
\]
They form the beginnings of reading intervals.
An interval with endpoints $B_k$ and $B_{k+1}$ will be referred to
as {\it cycle}.
Define also
\[
\{B_k':\, k \in \Z\} := \{T_n+\sigma_n:\, n \in \Z,\, \psi_n=1\}
\]
and again assume that
\[
\cdots < B_{-1}' < B_0' \le 0 < B_1' < \cdots
\]
These are the ends of reading intervals.
The two sequences, $\{B_k\}$ and $\{B_k'\}$, are interlaced:
between two successive elements of one sequence there is exactly one element 
of the other. See Figure \ref{fig:BRC}.
\begin{figure}[h]
\centering
\includegraphics[width=6cm]{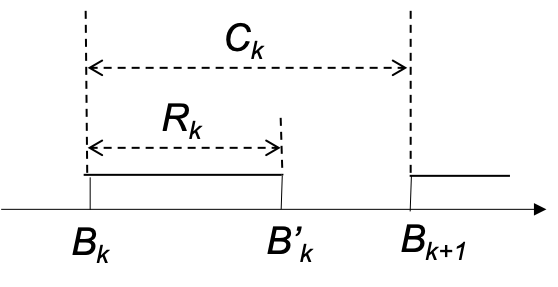}
\caption{\it The interval $[B_k, B_{k+1})$ is a cycle and the
subinterval $[B_k, B_k')$ is a reading interval.}
\label{fig:BRC}
\end{figure}
An interval with endpoints $B_k$ and $B_{k+1}$ is called a {\it cycle}.
We set 
\[
\mathbf C_k:=B_{k+1}-B_k\] 
for the cycle length.
The subinterval with endpoints $B_k$ and $B_k'$ is called a {\it reading interval}.
We set 
\[
\mathbf R_k:=B_k'-B_k
\]
for the reading length.

\subsection{The stationary framework and Palm probabilities}
\label{stfr}
Let $(\Omega, \FF, \widetilde \P)$ be a probability space endowed with a flow, i.e.,
a family of invertible measurable functions $\theta_t : \Omega \to \Omega$, $t \in \R$,
such that $\theta^{-1}_t$ are also measurable
and such that
\begin{equation}
\label{semig}
\theta_{t+s}= \theta_t \comp \theta_s, \quad s, t \in \R.
\end{equation}
Assume further that the flow preserves $\widetilde \P$, 
that is, 
\[
\widetilde \P \comp \theta_{t} = \widetilde \P, \quad t \in \R.
\]

Let $T_n, \sigma_n$ be random variables such that
the marked\footnote{A point process $\phi$
on a product space $S \times M$ is called $M$-marked (or just marked) if
$\phi(\{s\} \times M) \in \{0,1\}$ for all $s \in S$.
}
 point process $\sum_n \delta_{(T_n, \sigma_n)}$
is stationary, that is, 
\begin{equation}
\label{st1}
\bigg(\sum_n \delta_{(T_n, \sigma_n)}\bigg) \comp \theta_t
= \sum_n \delta_{(T_n-t, \sigma_n)}, \quad t \in \R.
\end{equation}
Note then that 
\[
A_t \comp \theta_s = A_{t+s}-s, \quad s,t \in \R.
\]
It follows that the arrival rate
\[
\lambda := \widetilde \E \sum_n \1_{0 \le T_n \le 1}
\]
is positive and finite.
Consider next a  acceptance policy as specified by 
the acceptance random variables $\chi_n$, $n \in \Z$, defined on $(\Omega, \FF)$.
We say that the system is in {\it steady-state} if, in addition to \eqref{st1}, 
\begin{equation}
\label{st2}
\bigg(\sum_n \delta_{(T_n, \sigma_n, \chi_n)}\bigg) \comp \theta_t
= \sum_n \delta_{(T_n-t, \sigma_n, \chi_n)}, \quad t \in \R.
\end{equation}
If the system is in steady-state then it follows from \eqref{st2}, \eqref{semig}
\eqref{psi} and \eqref{T'} that
\begin{equation}
\label{st3}
\bigg(\sum_n \delta_{(T_n, \sigma_n, \chi_n,\psi_n)}\bigg) \comp \theta_t
= \sum_n \delta_{(T_n-t, \sigma_n, \chi_n,\psi_n)}, \quad t \in \R,
\end{equation}
and, for all $s,t \in \R$,
\begin{gather*}
S_t \comp \theta_s = S_{t+s}-s,\quad
D_t \comp \theta_s = D_{t+s}-s,
\\
\alpha(s) \comp \theta_t = \alpha(t+s), \quad
\beta(s) \comp \theta_t = \beta(t+s), \quad
q(s) \comp \theta_t = q(t+s).
\end{gather*}
In general, it is not obvious that \eqref{st2} holds.
Of the four acceptance policies mentioned above, the pushout $\PO$ immediately
satisfies \eqref{st2} owing to  that
$\chi_n=1$
and $\psi_n=\1_{ T_{n+1}-T_n\ge \sigma_n}$ for all $n$. 
That \eqref{st2} holds is proved in
\cite[Section 5.3]{BB} and may require enlarging the probability space 
$(\Omega, \FF, \widetilde \P)$.

\begin{definition}
\label{defpalm}
We shall denote by $\P$ the Palm probability of $\widetilde \P$ with
respect to the point process $\arr=\sum_{n \in \Z} \delta_{T_n}$.
If \eqref{st2} holds we shall denote by 
$\Pstar$ the Palm probability of $\widetilde \P$ with
respect to the point process  $\sum_{k \in \Z} \delta_{B_k}$.
\end{definition}
For the notion of Palm probability see, e.g., 
Daley and Vere-Jones \cite[Chapter 13]{DV}, Kallenberg \cite{KALL}
and Baccelli and Br\'emaud \cite{BB}.
Formally, with $\BB$ denoting the class of Borel sets on $\R$,
the measure $\BB\ni C \mapsto \widetilde \E(\1_A \sum_n \1_{T_n \in C})$
is absolutely continuous, and hence differentiable, with respect to
the measure $\BB\ni C \mapsto \widetilde \E(\sum_n \1_{T_n \in C})$.
The value of the derivative at $0$ is precisely $\P(A)$.
The Palm probability $\Pstar (A)$ can be obtained in exactly the same manner.
However,
since $\{B_k\}$
is precisely the set of $T_n$ for which $\mathcal Z_n$ holds, it follows
that $\Pstar$ is obtained from $\P$ via elementary conditioning:
\begin{equation*}
\label{PstarP}
\Pstar = \P(\cdot | \mathcal Z_0).
\end{equation*}

The hierarchy of the three measures used in the paper is
\[
\widetilde \P \longrightarrow \P \longrightarrow \P^*
\]
Intuitively, one thinks of $\P$ is obtained from $\widetilde \P$ by conditioning that
a point of $(T_n)$ is at $0$ and $\P^*$ is obtained from $\P$ by conditioning
on that one of this point at $0$ is one of the points of $(B_k)$.
Hence $\P^*$ is obtained from $\widetilde \P$ as well 
by conditioning on both events. Hence if $A$ is an event such that
$\widetilde \P(A)=1$ then $\P(A)=1$ also
and if $\P(A)=1$ then $\P^*(A)=1$ also.
Integrals with respect to $\widetilde \P$, $\P$, $\Pstar$  are denoted by
$\widetilde \E$, $\E$, $\Estar$ respectively. 
Moreover, 
\begin{equation}
\label{BTone}
\P(T_0=0)=1, \quad \Pstar (B_0=T_0=0)=1.
\end{equation}
We denote by $\theta_{T_n}$
the map defined by $\theta_{T_n}(\omega) = \theta_{T_n(\omega)}(\omega)$. 
Then $\theta_{T_n}$, $n \in \Z$,
forms a discrete time flow that preserves $\P$.
In other words, $\P$-a.s., $\theta_{T_n} \comp \theta_{T_m} =
\theta_{T_{n+m}}$ for all $m,n \in \Z$ and
$\P \comp \theta_{T_n} = \P$ for all $n \in \Z$.
Similarly, $\Pstar$-a.s., $\theta_{B_k} \comp \theta_{B_\ell} 
= \theta_{B_{k+\ell}}$ for all $k, \ell \in \Z$ and $\Pstar \comp \theta_{B_k}
= \Pstar$ for all $k \in \Z$.

The $\P$-law of $(\tau_n,\sigma_n)$ does not depend on $n$.
In what follows, we let $(\tau, \sigma)$ be a generic random element whose
law is the same as the $\P$-law of $(\tau_0, \sigma_0)$.
The definition of Palm probability and the fact $\lambda>0$ implies that
\[
\E \tau = 1/\lambda < \infty.
\]
This is the minimal condition imposed by stationarity and thus it cannot be avoided.
It is important to note however
that we shall make no assumptions about finiteness of 
higher $\P$-moments of $\tau$.

Referring to Figure \ref{fig:BRC}, note that,
under $\Pstar$, all cycles have identical law and so do all reading
intervals.
We denote by $\mathbf C$  a {\it typical cycle length}, that is, a random
variable whose law is the $\Pstar$-law of the length of any cycle. 
Similarly, $\mathbf R$  denotes a {\it typical reading interval length}.

\section{Outline of  some of the results}
\label{sec:outline}
All results concern stationary processes. 
Denote by $\alpha_\PO$, $\alpha_\BL$ the AoI processes for the 
pushout and blocking systems, respectively.
Similarly, we let $\beta_\PO$, $\beta_\BL$ be the NAoI processes for the two systems.

\subsection{Stochastic decomposition/representation results}
These are obtained under the assumptions that, under the Palm measure $\P$, 
the $(\tau_i)$ are i.i.d.\ and independent
of the $(\sigma_i)$ which are also i.i.d. We refer to these assumptions
as being the i.i.d.\ (or renewal) assumptions.
When we say ``decomposition'' of (the law of) a random variable $X$ we mean,
as usual in applied probability and queueing theory, that $X \eqdist X_1 + X_2$
where $X_1$ and $X_2$ are independent random variables.
The following are obtained in Theorems \ref{alphaiid}, \ref{Hermes}, respectively.
Under $\wP$,
\begin{align*}
\alpha_\PO(t) &\eqdist \overline \tau + \bf R_\PO \\
\alpha_\BL(t) &\eqdist  \sigma + \bf \overline C_\BL
\end{align*}
Here, $\overline \tau$ is a random variable whose law is the law of the
stationary version of the interarrival time, $\bf R_\PO$  is distributed as the typical reading interval of the pushout system, 
and $\bf \overline C_\BL$ is  distributed as the stationary version of the typical cycle
of the blocking system.
We also obtain, in Theorems \ref{betaiid}, \ref{betagen2}, respectively,
the following representations:
\begin{align*}
(\beta_\PO(t)|\beta_\PO(t)>0) &\eqdist \bf C_\PO\\
\beta_\BL(t) \1_{\beta_\BL(t)>0} &\eqdist \beta_+(t).
\end{align*}
Here,  $\bf C_\PO$ is distributed as the typical cycle of the pushout system
and $\beta_+(t)$ is the NAoI process for an appropriately defined  
variant of the fully-blocking system: remove from
the system all undisturbed messages, that is, all messages that
arrive at an idle system and are such that no other messages
arrive while they are being processed.
Moreover, we find that the NAoI always has an atom at $0$.
This is obvious for $\beta_P$ because the it is $0$ when the processor
is idle, but less obvious for $\beta_\BL$. The last representation result
explains the appearance of an atom. For more discussion see Remark \ref{rembet}
of Section \ref{naoibl}.

\subsection{A guide to the subsequent analysis and results} 
We stress some points that will facilitate the reader in going through the analysis
of the pushout and blocking systems,
 Sections \ref{sec:fpo} and \ref{sec:fbl} below.
First of all, the reader should keep in mind 
the hierarchy of the three measures, $\widetilde \P$ (governing the stationary
system), $\P$ (Palm with respect to arrivals), and $\P^*$ (palm with respect to
the beginnings of cycles) should be kept in mind, as explained above.

Regarding the pushout system, the most general results are in Theorems \ref{alphagen} and \ref{betagen}:
\begin{align*}
&\wE F'(\alpha_\PO(0)) 
= \lambda \E\left[ F\bigg(\tau_{-1}+ \sum_{i=0}^{N-1} \tau_i+ \sigma_N\bigg)
- F(\sigma_N);\, \tau_{-1}> \sigma_{-1} \right],
\\
&\wE f(\beta_\PO(0))
= \lambda \E\left[
\sum_{i=0}^{N-1}\tau_i f\left(\sum_{j=-1}^{i-1} \tau_j\right)
+ \sigma_N f\left(\sum_{j=-1}^{N-1} \tau_j\right)
+ (\tau_N-\sigma_N) f(0);\, \tau_{-1} > \sigma_{-1}
\right].
\end{align*}
These are, in principle, expressions for the distributions of $\alpha_\PO(0)$
and $\beta_\PO(0)$ in steady-state because $F$ and $f$ are ``general'' functions
and everything on the right-hand sides of the equations depends solely on
the (joint) distribution of the infinite random sequence $(\tau_i, \sigma_i:\, i \in \Z)$.
In particular, $N$ is defined as 
$N= \inf\{\ell \ge 0:\, \tau_{\ell} \ge \sigma_{\ell}\}$ and denotes the index of the
first message, among the ones numbered $0,1,2,\ldots$, that is successful.
Note, in particular, that $N$ has a stopping time property and this, along with the 
fundamental probabilist's tool, the d\'ecoupage de L\'evy (Lemma \ref{lemddL}),
makes, under renewal assumptions,  the analysis and the obtaining of explicit formulas possible.

Regarding the blocking system, the most general results are 
formulas \eqref{Plato} and  \eqref{posvalue} of Theorems \ref{alphagen2} and \ref{betagen2} below. The formulas are more complicated due to the fact that the 
dynamics of the system and, in particular, the construction of the unique steady-state
depends on the infinite past. However, again, these formulas are again expressions
for the distributions of $\alpha_\BL(0)$ and $\beta_\BL(0)$. We point out that
the index $N$ appearing in them is now defined as 
$N = \inf\{\ell \ge 1:\, \tau_0+\cdots+\tau_{\ell-1} \ge \sigma_0\}$ and is chosen
so that it has the stopping time property.

Using renewal theory, we manage to turn these general formulas into explicit
results for the Laplace transforms of the quantities of interest. To do so, we
need to introduce several functionals of the processes which can be found by
solving fixed point (renewal equations). Sometimes, the Laplace transforms
can be inverted explicitly giving formulas for  densities. In particular,
this can be done when the random variables $(\tau_i)$ are i.i.d.\ exponential
and the $(\sigma_i)$ are also i.i.d.\ exponential and the two sequences
are independent. This, of course, is no surprise to the queueing theorist.
Finding just
the expectations of the AoI and NAoI can be done either via their
Laplace transforms or via the general formulas obtained via Palm calculus
by choosing specific functionals. 
We do whatever is quicker and obtain expectation
formulas that are summarized in Table \ref{tab:means} of the last 
section.  To the best of our knowledge, the GI/GI 
formulas are new and some of the rest are consistent with
\cite{Yates17}.

\section{The pushout system}
\label{sec:fpo}
The dynamics of the pushout system is quite simple:
every arriving message is admitted: $\chi_n=1$ for all $n \in \Z$.
The message arriving at $T_n$ is successful if and only if $T_n+\sigma_n
\le T_{n+1}$.
Thus
\[
\psi_n = \1_{\tau_n \ge \sigma_n}, \quad n \in \Z,
\]
where $\tau_n=T_{n+1}-T_n$ as in \eqref{taun}.
Since, for all $n$,   $\chi_n=1$ and
$\psi_n=\1_{\tau_n\ge \sigma_n}$, it follows from \eqref{T'} 
that the state process $q$ of \eqref{qdef} is alternatively given by
\[
q(t) = 
\begin{cases}
0, & T_n+\sigma_n \le t < T_{n+1}
~~\mbox{for some $n$}
\\
1, & \text{ otherwise}
\end{cases}.
\]
If $\P(\tau_0 < \sigma_0)=1$ then $\P(\tau_n < \sigma_n \text{ for all }n)=1$
and so $q$ is identically equal to $1$. This is an uninteresting case
resulting in infinite AoI and NAoI. 
We thus assume that
\begin{equation}
\label{Seneca}
\P(\tau_0 \ge \sigma_0) > 0,
\end{equation}
that is $\P(\psi_0=1)>0$. 
By the Poincar\'e recurrence theorem 
\cite[Theorem 7.3.4]{DURR}, there is a doubly-infinite subsequence 
$\psi_{n_k}$, $k \in \Z$, such that $\psi_{n_k}=1$ for all $k$,
$\P$-a.s.\ and $\widetilde \P$-a.s.
In other words,
$\inf\{n: \psi_n=1\}=-\infty$,
$\sup\{n:\psi_n=1\}=+\infty$, $\P$-a.s., and  hence $\widetilde \P$-a.s.
This implies that $\alpha, \beta$ are well-defined and finitely-valued processes.

It is easy to see that, for the pushout system,
the beginnings of cycles satisfy
\begin{gather*}
\{B_k: \, k \in \Z\} = \{T_n:\, n \in \Z, \psi_{n-1}=1\}.
\end{gather*}
We therefore have:
\begin{lemma}
The Palm probability
$\Pstar$ of Definition \ref{defpalm}
is the Palm probability of $\widetilde \P$ with
respect to the (stationary) point process 
\[
\sum_{n \in \Z} \psi_{n-1} \delta_{T_n}
\]
and
\begin{equation}
\label{PstarPpo}
\Pstar = \P(\cdot | \psi_{-1}=1) =  \P(\cdot | \tau_0 \ge \sigma_0) .
\end{equation}
\end{lemma}
In particular,
\begin{equation}
B_1 = \inf\{T_n: n \in \Z, T_n > 0, \psi_{n-1}=1\}, \quad
B_0 = \sup\{T_n: n \in \Z, T_n \le 0, \psi_{n-1}=1\}.
\label{BBB}
\end{equation}

\subsection{The age of information  for the pushout system}
To compute the law of $\alpha(0)$ we shall use the Palm inversion formula
\begin{equation}
\label{palminv}
\wE f(\alpha(0)) = \frac{\Estar \int_{B_0}^{B_1} f(\alpha(t))\, dt}
{\Estar (B_1-B_0)},
\end{equation}
where $f: \R \to \R$ is bounded and measurable or of constant sign and measurable.
The denominator is easy to compute:
\begin{equation}
\label{Bdiff}
\Estar (B_1-B_0) = \left(\wE \sum_n \psi_{n-1} \1_{0 < T_n < 1}\right)^{-1}
= \left(\lambda \E \int_\R \psi_{-1} \1_{0<t<1} \, dt\right)^{-1}
= \frac{1}{\lambda \P(\tau_0\ge \sigma_0)},
\end{equation}
where we used Campbell's formula.
By the non-triviality assumption \eqref{Seneca}, $\Estar(B_1-B_0)<\infty$.

We will need the following random integer below.
\begin{equation}
\label{Npo}
N: = \inf\{\ell \ge 0:\, \tau_{\ell} \ge \sigma_{\ell}\}
= \min\{\ell \ge 0:\, \tau_{\ell} \ge \sigma_{\ell}\}.
\end{equation}

\begin{theorem}
\label{alphagen}
Consider the pushout system under stationarity assumptions
and assume that \eqref{Seneca} holds.
Let $F: \R_+ \to \R$ be a bounded absolutely continuous function with a.e.\ derivative $F'$.
Then
\begin{equation}
\label{Socrates}
\wE F'(\alpha(0)) 
= \lambda \E\left[ F\bigg(\tau_{-1}+ \sum_{i=0}^{N-1} \tau_i+ \sigma_N\bigg)
- F(\sigma_N);\, \tau_{-1}> \sigma_{-1} \right],
\end{equation}
where $N$ is defined in \eqref{Npo}.
\end{theorem}

\begin{proof}
We have $N < \infty$ because of stationarity and hence the expression
in the brackets of \eqref{Socrates} makes sense.
Message $N$ is successful ($\psi_N=1$) and, by the first of \eqref{BBB}
and \eqref{BTone}, 
\[
B_0=T_0=0,\,
B_0' = T_N+\sigma_N,\,
B_1 = T_{N+1},  \quad \Pstar\text{-a.s.}
\]
To compute the integral in  the numerator of \eqref{palminv} 
we take a close look at the function $\alpha$
restricted on the interval $[B_0, B_1)=[T_0,T_{N+1})$.
Note that the only successful departures are precisely the points
$B_k'$ where reading periods end,
whereas the only successful arrivals are the last arrivals on a reading period.
If $T_0 \le t < T_N+\sigma_N$ then $D_t=B_{-1}'$
and so $S_{D_t}=S_{B_1'} = T_{-1}$, since $T_0=0$ initiates a reading period,
so the last successful arrival before this is the arrival that ended the 
previous reading period.
If $T_N +\sigma_N \le t < T_{N+1}$ then
$D_t=B_0'=T_N+\sigma_N$ and $S_{D_t} = S_{B_0'}=T_N$.
Thus,
\[
\alpha(t) =
\begin{cases}
t-T_{-1}, & T_0 \le t < T_N+\sigma_N
\\
t-T_N, & T_N+\sigma_N \le t < T_{N+1}
\end{cases}, \quad
\Pstar\text{-a.s.}
\]
Then, $\Pstar\text{-a.s.}$, $B_0=T_0=0$ (see \eqref{BTone}) and
\begin{align*}
\int_{B_0}^{B_1} f(\alpha(t)) \, dt
= \int_{T_0}^{T_{N+1}} f(\alpha(t))\, dt
&= \int_{T_0}^{T_N+\sigma_N} f(t-T_{-1})\, dt
+ \int_{T_N+\sigma_N}^{T_{N+1}} f(t-T_N)\, dt
\\
&= F(T_N+\sigma_N-T_{-1})-F(T_0-T_{-1})
+F(T_{N+1}-T_N)-F(\sigma_N),
\end{align*}
and thus, since $\Estar F(T_0-T_{-1}) = \Estar F(T_{N+1}-T_N)$,
\[
\Estar \int_{B_0}^{B_1} f(\alpha(t)) \, dt
= \Estar \left[ F\bigg(\tau_{-1}+ \sum_{i=0}^{N-1} \tau_i+ \sigma_N\bigg)
- F(\sigma_N)\right].
\]
We can rewrite  \eqref{Bdiff} as  
$\Estar (B_1-B_0) = 1/\lambda \P(\tau_{-1} \ge \sigma_{-1})$.
Dividing the last display by this expression and
using the relation \eqref{PstarPpo} between $\Pstar$ and $\P$ we arrive at
\eqref{Socrates}.
\end{proof}

At this level of generality it is not possible to have a more explicit formula.
However, given information about the law of the sequence $(\tau_n, \sigma_n)$,
$n \in \Z$, we can proceed further. 
For example, assuming that the  $\tau_n$, $n \in \Z$, is
independent of $\sigma_n$, $n \in \Z$, and both sequences have known laws
then a further simplification is possible.
If, in addition, the $\P$-law of one of the sequences is that of
i.i.d.\ exponential random variables then it is possible to elaborate further
and derive an almost closed-form formula.
\begin{theorem}
\label{alphaiid}
Consider the pushout system and
assume that $(\tau_n, \sigma_n)$, $n \in \Z$, is i.i.d.\ under $\P$
and such that $\E \tau_0 <\infty$ and $\P(\tau_0 \ge \sigma_0) > 0$.
Assume further that $\tau_n$ is independent of $\sigma_n$ for all $n$.
Then, for $u>0$,
\begin{equation}
\label{aLAP}
\wE e^{-u \alpha(0)} 
= \frac{1-\E e^{-u\tau}}{u \E\tau}\,
\frac{\E [e^{-u \sigma}; \tau \ge \sigma]}{1-\E[e^{-u \tau}; \tau < \sigma]}
\end{equation}
In particular, under $\wP$, $\alpha(0)$ is the sum of two independent
random variables:
\begin{equation}
\label{decomp1}
\alpha(0) \eqdist \overline \tau + \mathbf R , 
\end{equation}
where $\overline \tau$ is the stationary version of $\tau$
and $\mathbf R$ is a typical reading interval length.
\end{theorem}

\begin{corollary}
The $\wP$-distribution of $\alpha(0)$ is absolutely continuous.
\end{corollary}
To prove Theorem \ref{alphaiid}, 
we shall make use of the following elementary fact, often known
under the name ``d\'ecoupage de L\'evy''.
\begin{lemma}
\label{lemddL}
Let $X_1, X_2, \ldots$ be i.i.d.\ random elements in an arbitrary measurable space 
$(S, \SS)$ with common law $\mu$ and let  $B \in \SS$ have $\mu(B)>0$.
Let  $N=\inf\{n\ge1: X_n \in B\}$. Then
\begin{enumerate}
\item[(i)] 
$(X_1,\ldots,X_{N-1})$ is independent of $X_N$;
\item[(ii)]  $X_N$ has law $\mu(\cdot|B)$;
\item[(iii)] $\P(N=n) = \mu(S-B)^{n-1}\mu(B)$, $n \ge 1$.
\end{enumerate}
Moreover, the distribution of $(X_1, \ldots, X_N)$ can be expressed neatly as follows.
Let $X'', X'_1, X'_2, \ldots$ be {independent} random elements,
and independent of $N$,
such that
\[
\P(X'' \in \cdot) =\mu(\cdot|B), \quad 
\P(X_i'\in \cdot) = \mu(\cdot|S-B), \quad i=1,2,\ldots
\]
Then 
\[
(X_1, \ldots, X_N) \eqdist (X_1',\ldots,X'_{N-1}, X''),
\]
where, by definition, $(X_1',\ldots,X'_{N-1}, X'')=X''$ if $N=1$.
\end{lemma}
The proof is trivial and is thus omitted.

\begin{proof}[Proof of Theorem \ref{alphaiid}]
For fixed $u>0$, let $F(x) = e^{-u x}$, $x \ge 0$.
Then $F'(x)=-u e^{-ux}$ and $F(x_1+x_2)=F(x_1)F(x_2)$ for all $x_1, x_2 \ge 0$.
With a view towards applying Lemma \ref{lemddL} to the
sequence $(\tau_n,\sigma_n)$, $n \ge 0$,
let $B := \{(t,s) \in \R^2: t\ge s \ge 0\}$.
For simplicity,
let 
\[
p:= \P(\tau \ge  \sigma), \quad q =1-p.
\]
By \eqref{Socrates},
\begin{align*}
\wE F'(\alpha(0)) &= \lambda p\,
\Estar\left[ F\bigg(\tau_{-1}+ \sum_{i=0}^{N-1} \tau_i+ \sigma_N\bigg)
- F(\sigma_N)\right]
= \lambda p\, 
\E\left[ F\bigg(\tau''+ \sum_{i=0}^{N-1} \tau_i'+ \sigma''\bigg)
- F(\sigma'')\right],
\end{align*}
where $N, \tau'', \tau_1',\tau_2', \ldots, \sigma''$ are independent random variables
such that 
\begin{equation}
\label{newvars}
\text{$\P(N=n)=q^n p$,\quad
$\tau'' \eqdist (\tau | \tau > \sigma)$,\quad
$\sigma'' \eqdist (\sigma | \tau > \sigma)$,\quad
$\tau' \eqdist (\tau| \tau \le  \sigma)$.}
\end{equation}
Hence, letting $F(x)=e^{-ux}$ for some fixed $u>0$ we have
\begin{align*}
\wE F'(\alpha(0)) &=  \lambda p\, \E \left\{ F(\tau'')  F(\sigma'')  \prod_{i=0}^{N-1} F(\tau_i') - F(\sigma'') \right\}
= \lambda p\, \E F(\sigma'') \,
 \left\{ \E F(\tau'')\,  \E[( \E F(\tau'))^N] - 1 \right\}
\\
&= \lambda p\, \E F(\sigma'') \,
 \left\{ \E F(\tau'')\,  \frac{p}{1-q \E F(\tau')} - 1 \right\}
=  \lambda p\,  \frac{ \E F(\sigma'') \, (\E F(\tau)-1)}{1-q \E F(\tau')},
\end{align*}
whence, after a little algebra, we obtain \eqref{aLAP}:
\begin{align*}
- u \wE e^{-u \alpha(0)}
&= \lambda (\E e^{-u\tau}-1)\, 
\frac{p \E e^{-u \sigma''}}{1-q \E e^{-u \tau'}}
= \lambda (\E e^{-u\tau}-1)\, 
\frac{\E[ e^{-u \sigma}; \tau \ge  \sigma]}{1-\E[e^{-u \tau}; \tau <\sigma]}.
\end{align*}
To prove \eqref{decomp1} note that the first term in \eqref{aLAP}
equals  $\frac{1-\E e^{-u\tau}}{u \E\tau}$ is equal to
$\E e^{-u \overline \tau}$.
So $\alpha(0) \eqdist \overline \tau+Y$ where $Y$ is an 
independent random variable
whose Laplace transform is the second term in \eqref{aLAP}:
\begin{equation}
\label{Ylap}
\E e^{-u Y} 
= \frac{\E [e^{-u \sigma}; \tau \ge \sigma]}{1-\E[e^{-u \tau}; \tau < \sigma]}.
\end{equation}
Recalling that $N$ is the index of the first successful  arrival  after
the origin, we see that, again after a little algebra involving 
a geometric series,
\begin{equation}
\label{Ylaplap}
\E e^{-u (T_N+\sigma_N) } 
= \E \sum_{n=0}^\infty e^{-u(\tau_0+\cdots+\tau_{n-1}+\sigma_n)}\,
\1_{\tau_0 < \sigma_0, \ldots, \tau_{n-1} < \sigma_{n-1}, \tau_n\ge \sigma_n}
= \frac{\E [e^{-u \sigma}; \tau \ge \sigma]}{1-\E[e^{-u \tau}; \tau < \sigma]}.
\end{equation}
This shows that  $\E e^{-u Y} = \E e^{-u (T_N+\sigma_N) }$
for all $u>0$, and thus
\[
Y \eqdist T_N+\sigma_N.
\]
But, $\Pstar$--a.s.,  $T_N+\sigma_N=B_0'-B_0 \eqdist \mathbf R$.
\end{proof}

\begin{remark}
We may decompose $\alpha(0)$ in a different way.
Rearranging terms in the $\wP$-Laplace transform of $\alpha(0)$ we have
\[
\widetilde \E e^{-u \alpha(0)}
= \E e^{-u \sigma''} \,  \frac{\lambda p}{u}
\frac{1-\E e^{-u\tau}}{1-q\E e^{-u \tau'}},
\]
which implies that there is a second decomposition for the law of $\alpha(0)$:
\[
\alpha(0) \eqdist \sigma'' + Z, 
\]
where $\sigma''$ and $Z$ are independent random variables,
with $\sigma''$ having the law of $\sigma$ conditional on $\tau\ge \sigma$
and $Z$ having Laplace transform $({\lambda p}/{u})
({1-\E e^{-u\tau}})/({1-q\E e^{-u \tau'}})$.
\end{remark}

\begin{corollary}\label{cor:AoI-PO-mean}
Under the assumptions of Theorem \ref{alphaiid}, we have
\begin{equation}
\label{a1}
\wE \alpha(0) = \frac{\E \tau^2}{2\E \tau} + \frac{\E \tau \wedge \sigma}{
\P(\tau \geq \sigma)}.
\end{equation}
\end{corollary}
\begin{proof}
Look at \eqref{decomp1}. We have
$\E \overline \tau = \E \tau^2/2\E \tau$
and 
\[
\E \mathbf R
= \E(T_N+\sigma_N) = \frac{\E \tau \wedge \sigma}{p}.
\]
\end{proof}

\begin{corollary}
\label{coroxxx}
Under the assumptions of Theorem \ref{alphaiid},
and if, in addition, the variables $\sigma_n$ are exponential with rate $\mu$,
then, under $\wP$,
\[
\alpha(0) \eqdist \overline \tau + \frac{\mathbf e}{\mu},
\]
where $\e$ is a rate-1 exponential random variable, 
independent of $\overline \tau$
and so
\[
\wE \alpha(0) = \frac{\E \tau^2}{2\E \tau} + \frac{1}{\mu}.
\]

\end{corollary}
\begin{proof}
We use \eqref{decomp1}.
We just have to show that the reading interval length $\mathbf R$ is exponential with rate $\mu$. Since
\begin{align*}
\E [e^{-u \sigma}; \tau \ge \sigma]
&= \E \int_0^\tau e^{-u s} \mu e^{-\mu s} ds
=\mu  \E \int_0^\tau e^{-(u+\mu) s} ds
= \frac{\mu}{u+\mu} [1-\E e^{-(u+\mu)\tau}],
\\
\E[e^{-u \tau}; \tau < \sigma]
&= \E e^{-u \tau} \P(\sigma \ge \tau|\tau)
= \E  e^{-u \tau}  e^{-\mu \tau}
= \E  e^{-(u+\mu) \tau},
\end{align*}
we have, from \eqref{Ylap}, that the Laplace transform of $\mathbf R$ is
\[
\E e^{-u \mathbf R} =\frac{\E [e^{-u \sigma}; \tau \ge \sigma]}{1-\E[e^{-u \tau}; \tau < \sigma]}
=\frac{\frac{\mu}{u+\mu} [1-\E e^{-(u+\mu)\tau}]}
{1-\E  e^{-(u+\mu) \tau}} = \frac{\mu}{u+\mu}.
\]
\end{proof}

\begin{corollary}
\label{cor:four}
Under the assumptions of Theorem \ref{alphaiid},
and if, in addition, the variables $\tau_n$ are exponential with rate $\lambda$,
then
{
\[
\widetilde \E e^{-u \alpha(0)}
= \frac{\lambda \E e^{-(\lambda+u)\sigma}}{u+\lambda \E e^{-(\lambda+u)\sigma}},
\quad
\wE \alpha(0) = \frac{1}{ \lambda\, \E e^{-\lambda \sigma}}.
\]
}
\end{corollary}
\begin{proof}
Since $\tau$ is exponential we have $\overline \tau \eqdist \tau$ and so
\[
\E e^{-u \overline \tau} = \E e^{-u \tau} = \frac{\lambda}{u+\lambda}.
\]
Using \eqref{Ylap}, we have
\[
\E e^{-u \mathbf R} 
= \frac{(u+\lambda) \E e^{-(u+\lambda)\sigma}}
{u+\lambda \E e^{-(u+\lambda)\sigma}}.
\]
Equation \eqref{aLAP} says that the Laplace transform of $\alpha(0)$
is the product of the last two displays and so this derives
the first formula.
Next use \eqref{a1}.
Since
\[
\E \tau \wedge \sigma = \frac{1}{\lambda} (1-\E e^{-\lambda\sigma}),
\quad \P(\tau> \sigma) = \E e^{-\lambda \sigma},
\]
we have
\[
\wE\alpha(0) = \frac{1}{\lambda} + \frac{1}{\lambda} \cdot
\frac{1-\E e^{-\lambda\sigma}}{\E e^{-\lambda \sigma}}
= \frac{1}{\lambda \E e^{-\lambda \sigma}}.
\]
\end{proof}

Finally, a direct consequence of either of the above corollaries is:
\begin{corollary}
\label{coroa}
If the $\tau_n$ are i.i.d.\ exponential with rate $\lambda$,
if the $\sigma_n$ are i.i.d.\ exponential with rate $\mu$, and 
if the two sequences are independent, then, under $\wP$,
\[
\alpha(0) \eqdist \frac{\tt \e_1}{\lambda} + \frac{\tt \e_2}{\mu},
\]
where $\tt \e_1, \tt \e_2$ are two independent unit-rate exponential
random variables.
\end{corollary}

\subsection{The new age of information for the pushout system}\label{sec:NAoI-PO}
Recall that $\beta(t)=A_t - S_{D_t}$. Under $\wP$, the law of
$\beta(t)$ is independent of $t$.
\begin{lemma}
\label{rematom}
The $\wP$-law of $\beta(t)$ has a nontrivial atom at $0$.
\end{lemma}
\begin{proof}
Indeed,
\[
\wP (\beta(t)=0) = \wP (A_t=S_{D_t})
= \wP(q(t)=0) > 0.
\]
The latter is positive because of the non-triviality assumption \eqref{Seneca}.
\end{proof}

\begin{theorem}
\label{betagen}
Consider the pushout system under stationarity assumptions
and assume that \eqref{Seneca} holds.
Let $f: \R_+ \to \R$ be a  measurable function
that is bounded or nonnegative.
Then
\begin{equation}
\label{betapo}
\wE f(\beta(0))
= \lambda \E\left[
\sum_{i=0}^{N-1}\tau_i f\left(\sum_{j=-1}^{i-1} \tau_j\right)
+ \sigma_N f\left(\sum_{j=-1}^{N-1} \tau_j\right)
+ (\tau_N-\sigma_N) f(0);\, \tau_{-1} > \sigma_{-1}
\right],
\end{equation}
where $N$ is as in Theorem \ref{alphagen}.
\end{theorem}

\begin{proof}
We use again the Palm inversion formula
\begin{equation}
\label{palminvbeta}
\wE f(\beta(0)) = \frac{\Estar \int_{B_0}^{B_1} f(\beta(t))\, dt}
{\Estar(B_1-B_0)},
\end{equation}
where the notation is as before.
We now have
\[
\beta(t) = A_t-S_{D_t}
= \begin{cases}
T_i - T_{-1}, & T_0 \le T_i \le t < T_{i+1} \le T_N+\sigma_N,\, i \ge 0,
\\
0, & T_N+\sigma_N \le t < T_{N+1}
\end{cases},\quad \Pstar\text{-a.s.}
\]
Hence the integral in \eqref{palminvbeta} is
\begin{align*}
\int_{B_0}^{B_1} f(\beta(t)) \, dt
&= \int_{T_0}^{T_{N+1}} f(\beta(t))\, dt
\\
& = \sum_{i: T_0 \le T_i < T_{i+1} \le T_N} 
\int_{T_i}^{T_{i+1}} f(T_i - T_{-1}) dt
+ \int_{T_N}^{T_N+\sigma_N} f(T_N - T_{-1}) dt
+ \int_{T_N+\sigma_N}^{T_{N+1}} f(0) dt
\\
&=
\sum_{i=0}^{N-1}\tau_i f(T_i-T_{-1}) 
+ \sigma_N f(T_N - T_{-1})
+ (\tau_N-\sigma_N) f(0).
\end{align*}
Substitute this into \eqref{palminvbeta} and use
$\Estar (B_1-B_0) = 1/\lambda \P(\tau_{-1} \ge \sigma_{-1})$
to obtain \eqref{betapo}.
\end{proof}

\begin{corollary}[Continuation of Lemma \ref{rematom}]
The atom of $\beta(0)$ at $0$ has value
\begin{equation}
\label{0atom}
\wP (\beta(0)=0) =\lambda\, 
\E[(\tau_N-\sigma_N);\, \tau_{-1} > \sigma_{-1}].
\end{equation}
\end{corollary}
\begin{proof}
Let, in \eqref{betapo}, $f(x) := \1_{x=0}$.
Since all the $\tau_n$ and $\sigma_n$ are nonzero with probability $1$,
\eqref{0atom} follows.
\end{proof}

\begin{theorem}
\label{betaiid}
Consider the pushout system and
assume that $(\tau_n, \sigma_n)$, $n \in \Z$, is i.i.d.\ under $\P$
and such that $\E \tau_0 <\infty$ and $\P(\tau_0 \ge \sigma_0) > 0$.
Assume further that $\tau_n$ is independent of $\sigma_n$ for all $n$.
Then  the $\wP$-law of $\beta(0)$ can be described as
\begin{equation}
\label{betacases}
\beta(0) \eqdist
\begin{cases}
0 , & \text{with probability } \displaystyle \frac{\E(\tau-\sigma)^+}{\E\tau}
\\[2mm]
\mathbf C, & \text{with probability } \displaystyle \frac{\E \tau \wedge \sigma}{\E \tau}
\end{cases},
\end{equation}
where $\mathbf C$ has the distribution of a typical cycle length;
\begin{equation}
\label{typicalcycle}
\E e^{-u \mathbf C} = \Estar e^{-u  (B_1-B_0)}
= \frac{\E[e^{-u\tau}; \tau> \sigma]}{1-\E[e^{-u\tau};  \tau\le \sigma]},
\end{equation}
In particular,
\begin{equation}
\label{betamean}
\wE \beta(0) = \frac{\E \tau \wedge \sigma}{\P(\tau \ge \sigma)}.
\end{equation}
\end{theorem}
\begin{proof}
Using \eqref{0atom} and independence,
\begin{align*}
\wP( \beta(0)=0) &= \lambda\, \E(\tau_N-\sigma_N)
\,\P(\tau_{-1}> \sigma_{-1}).
\end{align*}
By Lemma \ref{lemddL} and \eqref{newvars}, we further have
\begin{align*}
\wP( \beta(0)=0) &=\lambda\, \E(\tau''-\sigma'') \, \P(\tau> \sigma)
\\
&= \lambda\,\E(\tau-\sigma | \tau> \sigma) \, \P(\tau>\sigma)
\\
&=\lambda\, \E (\tau-\sigma)^+.
\end{align*}
This proves the upper part of \eqref{betacases}.
To prove the lower part notice, from \eqref{betapo}, 
\begin{align*}
\wE[ f(\beta(0)); \beta(0)>0]
&= \lambda \, \P(\tau_{-1} > \sigma_{-1})\,\E\left[
\sum_{i=0}^{N-1}\tau_i f\left(\sum_{j=-1}^{i-1} \tau_j\right)
+ \sigma_N f\left(\sum_{j=-1}^{N-1} \tau_j\right)
 \bigg| \tau_{-1} > \sigma_{-1}
\right]
\\
&= \lambda  p\, 
\E\left[ 
\sum_{i=0}^{N-1}
\tau_i'  f\left(\tau_{-1}''+\sum_{j=0}^{i-1} \tau_j'\right)
+\sigma'' f\left(\tau_{-1}''+\sum_{j=0}^{N-1} \tau_j'\right)
\right],
\end{align*}
where we used Lemma \ref{lemddL} and the definitions \eqref{newvars}.
Next, let $f(x) = e^{-ux}$ and write the above as
\begin{align*}
\wE[ f(\beta(0)); \beta(0)>0]
&=\lambda  p\, (\E f(\tau''))\,
\E\left[\sum_{i=0}^{N-1} (\E \tau') (\E f(\tau'))^i
+(\E \sigma'') (\E f(\tau'))^N
\right]
\\
&=\lambda  p\, (\E f(\tau''))\,
\left[
(\E \tau') \, \E \left(\frac{1-(\E f(\tau'))^N}{1-\E f(\tau')}\right)
+(\E \sigma'') (\E f(\tau'))^N
\right]
\\
&=\lambda  p\, (\E f(\tau''))\,
\left[
\frac{\E \tau'}{1- \E f(\tau')}
\left(1-\frac{p}{1-q \E f(\tau')} \right)
+ (\E \sigma'') \frac{p}{1-q \E f(\tau')}
\right]
\\
&=\lambda  p\, (\E f(\tau''))\, \frac{q \E \tau' + p \E \sigma''}
{1-q \E f(\tau')}
= \lambda  p\, (\E f(\tau''))\, \frac{\E \tau \wedge \sigma}{1-q \E f(\tau')},
\end{align*}
that is precisely the lower part of \eqref{betacases}.
The last equality in \eqref{typicalcycle} is easily verified along the same lines.
To finally show \eqref{betamean} just note that 
\[
\wE \beta(0) = \frac{\E \tau \wedge \sigma}{\E \tau}\, \E C
= \frac{\E \tau \wedge \sigma}{\E \tau}\, \frac{\E \tau}{\P(\tau > \sigma)}.
\]
\end{proof}
\begin{remark}\label{rem:integrable-tau}
Notice that $\beta$ does not suffer from the same drawback as $\alpha$
when $\tau^2$ is not integrable. Indeed, here, under the condition
$\E \tau< \infty$ we have $\wE \beta(0) \le 1$, regardless of the variance
of $\tau$.
\end{remark}

\begin{corollary}
\label{cor:seven}
Let the assumptions of Theorem \ref{betaiid} hold true.
\\
(i)
If the variables $\tau_n$ are exponential with rate $\lambda$,
then
\[
\widetilde \E e^{-u \beta(0)}
=1- \frac{u(1-\E e^{-\lambda \sigma})}{u+\lambda \E e^{-(\lambda+u) \sigma}},
\qquad
\wE \beta(0) = \frac{1}{\lambda \E e^{-\lambda \sigma}}-\frac{1}{\lambda}.
\]
(ii)
If the variables $\sigma_n$ are exponential with rate $\mu$,
then
\[
\widetilde \E e^{-u \beta(0)}
=1- \frac{1- \E e^{-\mu\tau}}{\mu \E\tau}
\, \frac{1-\E e^{-u\tau}}{1-\E e^{-(u+\mu)\tau}},
\qquad
\wE \beta(0) = \frac{1}{\mu}.
\]
(iii)
If the $\tau_n$ are  with rate $\lambda$,
and the $\sigma_n$ are exponential with rate $\mu$ then, under $\wP$,
\[
\beta(0) \eqdist 
\begin{cases}
0, & \text{ with probability } \frac{\mu}{\lambda+\mu}
\\
\frac{\mathbf e_1}{\lambda}+\frac{\mathbf e_2}{\mu}, & \text{ with probability } \frac{\lambda}{\lambda+\mu}
\end{cases},
\qquad
\wE \beta(0) = \frac{1}{\mu},
\]
where $\mathbf e_1$, $\mathbf e_2$ are two independent unit-rate exponential
random variables.
\end{corollary}

\section{The blocking system}
\label{sec:fbl}
The blocking system is defined by the requirement that only those messages
for which $\mathcal Z_n$ holds are admitted. The remaining ones are 
immediately rejected (blocked). 
It is well-known that if
\begin{equation}
\label{SenecaB}
\P(\sup_{i \le -1} (\sigma_i +T_i) \le 0) > 0
\end{equation}
then the system admits a unique steady-state,
see \cite[Chapter 2, Section 5.2]{BB}. Under this condition, \eqref{st3} holds.

We have $\psi_n=\chi_n$ for all $n \in \Z$ (a message is successful if and only
if it is admitted) and
\begin{equation}
\psi_n \text{ is a measurable function of }
(\tau_m, \sigma_m:\, m \le n-1).
\label{psim}
\end{equation}
Recall that we use letters $B_k$, $B_k'$ for the beginnings and ends of reading periods,
respectively.  In other words,
\begin{gather*}
\{B_k: \, k \in \Z\} = \{T_n:\, n \in \Z, \psi_{n}=1\},
\\
\{B_k': \, k \in \Z\} = \{T_n+\sigma_n:\, n \in \Z, \psi_{n}=1\}.
\end{gather*}
Therefore the Palm probability $\P^*$ of $\widetilde \P$ with respect to $\{B_k\}$ admits
a simpler representation:
\begin{lemma}
$\Pstar$ is  the Palm probability of $\widetilde \P$ with
respect to the (stationary) point process 
\[
\sum_{n \in \Z} \psi_{n} \delta_{T_n} 
\]
and
\begin{equation}
\label{PstarPbl}
\Pstar = \P(\cdot | \psi_{0}=1). 
\end{equation}
\end{lemma}

Recalling that $\{B_k\}$ and $\{B'_k\}$ are interlaced sequences let us
compute the quantities $S_t$ (last successful arrival before $t$), 
$D_t$ (last successful departure before $t$), 
and $S_{D_t}$ (last successful arrival before $D_t$) depending whether $t$ falls in
a reading interval (that is, between $B_k$ and $B_k'$ for some $k$)
or not (that is, between $B_k'$ and $B_{k+1}$ for some $k$).
Since $\{B_k\}$ is the totality of successful arrivals, we have that,
for all $k \in \Z$,
\[
B_k\le t < B_{k+1} \Rightarrow S_t = B_k.
\]
Since $\{B_k'\}$ is the totality of successful departures, we have that,
for all $k \in \Z$,
\[
B_k' \le t < B_{k+1}' \Rightarrow D_t = B_k'.
\]
It then follows that,
for all $k \in \Z$,
\begin{equation}
\label{onetwo}
S_{D_t} =
\begin{cases}
B_{k-1}, & \text{ if } B_k \le t < B_k'
\\
B_k , & \text{ if } B_k' \le t < B_{k+1}
\end{cases}.
\end{equation}

\subsection{The age of information  for the blocking system}
We shall use the Palm inversion formula \eqref{palminv} for the process
$\alpha(t)=t-S_{D_t}$, $t \in \R$, for the blocking system.
By Campbell's formula we have that the denominator of \eqref{palminv} is
\begin{equation}
\label{Binv}
\Estar(B_1-B_0) 
= \frac{1}{\lambda \P(\psi_0=1)},
\end{equation}
however, unlike in the pushout system, the probability in the denominator
depends on the full distribution and the dynamics of the system and so
it does not admit an explicit form without further assumptions.
In what follows, let
\begin{equation}
\label{Nbl}
N := \inf\{\ell \ge 1:\, \tau_0+\cdots+\tau_{\ell-1} \ge \sigma_0\}.
\end{equation}

\begin{theorem}
\label{alphagen2}
Consider the blocking system under stationarity assumptions
and assume that \eqref{SenecaB} holds.
Let $f$  be bounded and measurable or locally integrable and nonnegative
function and let $F$ be such that $F'=f$.
Then
{
\begin{align}
\wE f(\alpha(0)) 
&= \lambda \, \E[ F(T_N + \sigma_N) - F(\sigma_N) ;\,\psi_0=1]
= \frac{\E[ F(T_N + \sigma_N)- F(\sigma_N) | \psi_0=1]}{\E[T_N | \psi_0=1]},
\label{Plato}
\end{align}
}
where $N$ is defined by \eqref{Nbl}.
\end{theorem}

\begin{proof}
Under $\Pstar$, message $0$ is successful (admitted) and  $N$ is the first successful
(admitted) message after that. Note that $N<\infty$.
Thus,
\begin{equation}
\label{B1TN}
B_1=T_N, \quad \Pstar\text{-a.s.}
\end{equation}
Note also that, with $\arr = \sum_{n \in \Z} \delta_{T_n}$,
\begin{equation}
\label{Na}
N= \arr([0,\sigma_0]) = \sum_{n=0}^\infty \1_{T_n \le \sigma_0}, 
\quad \P\text{-a.s. and (hence) } \Pstar\text{-a.s.}
\end{equation}
By \eqref{onetwo},
and since $B_0' = T_0+\sigma_0$, $\Pstar$-a.s.,
the function $\alpha$ on $[B_0, B_1)$ is given by
\begin{equation*}
\label{adef}
\alpha(t)
= t-S_{D_t}=
\begin{cases}
t-B_{-1}, & T_0 \le t < T_0+\sigma_0
\\
t-B_0, & T_0+\sigma_0 \le t < T_N
\end{cases},\quad
\Pstar\text{-a.s.}
\end{equation*}
Hence, for  functions $f, F$ as in the theorem statement, with $F'=f$, 
\begin{align*}
\int_{B_0}^{B_1} f(\alpha(t)) \, dt
= \int_{T_0}^{T_{N}} f(\alpha(t))\, dt
&= \int_{T_0}^{T_0+\sigma_0} f(t-B_{-1})\, dt
+ \int_{T_0+\sigma_0}^{T_{N}} f(t-B_0)\, dt
\\
&= F(B_0-B_{-1}+\sigma_0)-F(B_0-B_{-1})+F(B_1-B_0)-F(\sigma_0),\quad
\Pstar\text{-a.s.},
\end{align*}
and thus, since $\Estar F(B_0-B_{-1})= \Estar F(B_1-B_0)$,
\begin{align*}
\Estar \int_{B_0}^{B_1} f(\alpha(t)) \, dt
&= \Estar F(B_0-B_{-1}+\sigma_0)- \Estar F(\sigma_0)
\\
&= \Estar F(B_1-B_0+\sigma_N) - \Estar F(\sigma_N).
\end{align*}
Here we used the fact that $\Pstar$ is preserved by $\theta_{B_k}$
for all $k \in \Z$.
Taking into account \eqref{palminv}, \eqref{Binv} and
\eqref{B1TN},  we can conclude.
\end{proof}
\begin{remark}
Note that, since there is no ready-made expression
for $\P(\psi_0=1)$, the second formula in \eqref{Plato}
turns out to be more useful for further computations.
\end{remark}

We now introduce
\begin{equation}
\arr(t) := \inf\{\ell \ge 0:\, T_\ell \ge t\}, \quad t \ge 0,
\label{Nblo}
\end{equation}
so that the variable $N$ defined by \eqref{Nbl} is simply the value of $\arr(t)$ for
$t=\sigma_0$:
\[
\arr(\sigma_0)=N.
\]
Note that $\arr(t)$ is left-continuous at all $0<t<\infty$ with
$zarr(0)=0$ and $\arr(0+)=1$. 
Since $\arr = \sum_{n \in \Z} \delta_{T_n}$,
we have
\[
\arr(t) = \arr([0,t))
= 1 + \arr((0,t)), \quad t \ge 0.
\]
Remembering that $\P$ is a Palm probability and $\P(T_0=0)=1$, define
\begin{equation}
\label{funU}
U(t) :=\E \arr(t)
= \sum_{n=0}^\infty \P(T_n < t), \quad t \ge 0.
\end{equation}
If the $\tau_n$ are i.i.d., then $U$ is known as $0$-potential function
(if $T_0, T_1, T_2, \ldots$ is thought of as a random walk)
or renewal function
(if $T_0, T_1, T_2, \ldots$ are thought of as the points of a renewal process).
We have that $U$ is left-continuous on $[0,\infty)$ with $U(0)=0$, $U(0+)=1$.
We shall deal with the renewal case next.
We will also need the definition
\begin{equation}
\label{funWf}
W(f,t) := \E f(T_{\arr(t)}), \quad t \ge 0,
\end{equation}
where $f$ is an appropriate function for which the expectation exists.
In particular, with $f(x)=e^{-ux}$ for some $u>0$, we let
\begin{equation}
\label{funW}
W_u(t) = \E e^{-u T_{\arr(t)}},
\end{equation}
and with $f(x) = x^p$ for some $p>0$, we let
\[
M_p(t) = \E T_{\arr(t)}^p.
\]
The following result gives the Laplace transform of the $\wP$-marginal of $\alpha(t)$
in terms of functions that can be computed as unique solutions to 
fixed-point equations.

\begin{theorem}
\label{Hermes}
Consider the blocking system and
assume that $(\tau_n, \sigma_n)$, $n \in \Z$, is i.i.d.\ under $\P$
and such that $\E \tau_0 <\infty$ and $\P(\tau_0 \ge \sigma_0) > 0$.
Assume further that $\tau_n$ is independent of $\sigma_n$ for all $n$.
Then, for $u>0$,
\begin{equation}
\label{alphaBL}
\wE e^{-u \alpha(0)} 
= \E e^{-u\sigma} \cdot \frac{1-\E e^{-u T_N}}{u \E T_N}
= \E e^{-u\sigma} \cdot \frac{1-\E W_{u}(\sigma) }{u\, \E \tau\, \E U(\sigma)},
\end{equation}
where $U$ and $W_{u}$ are the unique solutions to
the fixed-point equations
{
\begin{align}
U(t)  &= 1 + \int_{(0,t]} U(t-x) \P(\tau \in dx)  
\label{recU}
\\
W_{u}(t)  &= \int_{(t,\infty)} e^{-ux}\P(\tau \in dx)+ \int_{(0,t]}  W_{u}(t-x) e^{-ux} \P(\tau \in dx).
\label{recW}
\end{align}}
In particular, under $\wP$, $\alpha(0)$ is the sum of two independent random variables:
\begin{equation}
\label{alphadec}
\alpha(0) \eqdist \sigma + \overline {T_N},
\end{equation}
where $\overline {T_N}$ is the stationary version of $T_N$.
\end{theorem}

\begin{proof}
Observe first that $\P(\tau_0 \ge \sigma_0)>0$ implies  (by the ergodic theorem)
\eqref{SenecaB}
and hence a unique steady-state version exists.
Using the fact that $\psi_n$ is a measurable function of the variables
$\tau_m, \sigma_m$ with $m \le n$ [see \eqref{psim}] we write \eqref{Plato} as
\begin{equation}
\label{aagain}
\wE F'(\alpha(0)) 
= \frac{\E[ F(T_N + \sigma_N)- F(\sigma_N)]}
{\E T_N },
\end{equation}
with $N= \inf\{\ell \ge 1:\, \tau_0+\cdots+\tau_{\ell-1} 
\ge \sigma_0\}$, $\P\text{-a.s.}$
Since $N-1=\inf\{i \ge 0:\, \tau_0+\cdots+\tau_i \ge \sigma_0\}$,
it follows that $N-1$ is a stopping time with respect to $\mathscr A_i$, $i \ge 0$,
where $\mathscr A_i$ is the $\sigma$-algebra generated by $(\sigma_0, \tau_0,
\ldots,\tau_i)$.
Let $F(x) = e^{-ux}$. Then
\[
\E[ F(T_N + \sigma_N)- F(\sigma_N)]
= \E[ F(T_N) F(\sigma_N) - F(\sigma_N)]
= [\E F(T_N) -1]\, \E F(\sigma_N),
\]
where the last equality needs that $N-1$ is a stopping time.
Noting that $\E F(\sigma_N) = \E F(\sigma)$
we obtain the first equality in \eqref{alphaBL} from which
decomposition \eqref{alphadec} follows at once.
\\
For the last equality of \eqref{alphaBL} we have
\begin{multline}
\E\ T_N
= \E\sum_{i=0}^{N-1} \tau_i 
= \E  \sum_{i=0}^\infty \tau_i\, \1_{T_i \le \sigma_0}
=  \sum_{i=0}^\infty (\E \tau_i) \P(T_i \le \sigma_0)
= (\E \tau) \, \sum_{i=0}^\infty \P(T_i \le \sigma_0)
= (\E \tau) \, \E U(\sigma),
\label{ETN}
\end{multline}
and,
\begin{equation}
\label{EuTN}
\E e^{-u T_N} = \E e^{-u T_{\arr(\sigma_0)}}
=  \E\, \E[ e^{-u T_{\arr(\sigma_0)} } |\sigma_0]
=  \E\, W_{u}(\sigma_0).
\end{equation}
\\
Equation \eqref{recU} is the  renewal equation from standard renewal theory.
To obtain\eqref{recW} we write
\begin{equation*}
W(f,t) = \E f(T_{\arr(t)})=
\E [f(T_{\arr(t)});{\red t < \tau_0}] + \E [f(T_{\arr(t)}); {\red t \ge \tau_0}].
\end{equation*}
If $t < \tau_0$ then $\arr(t)=1$, $T_{\arr(t)}=T_1=\tau_0$, $\P$-a.s.
If $t\ge \tau_0$ and $\tau_0=x$ then
$T_{\arr(t)} \eqdist x+T_{\arr(t-x)}$, under $\P$. Set 
\[
\Phi_t:=T_{\arr(t)}.
\]
If $\tau$ is independent of
$(\Phi_t)$ we have 
\[
T_{\arr(t)} \eqdist \tau+\Phi_{t-\tau}
\]
and so 
\[
f(T_{\arr(t)}) \1_{\tau_0 \le t} \eqdist f(\tau+\Phi_{t-\tau}) \1_{\tau\le t}.
\]
Hence
\begin{equation}
\label{FFF}
W(f,t)= \E[f(\tau);\, \tau >t] + \E[f(\tau+\Phi_{t-\tau});\, \tau \le t].
\end{equation}
Letting $f(x)=e^{-ux}$ we further have
$F(\tau+\Phi_{t-\tau}) = e^{-u\tau} e^{-u \Phi_{t-\tau}}$
and so
\[
\E[e^{-u(\tau+\Phi_{t-\tau})};\, \tau \le t]
= \E[e^{-u\tau} \E(e^{-u \Phi_{t-\tau}}|\tau) \1_{\tau \le t}]
= \E[e^{-u\tau} W_{u}(t-\tau) \1_{\tau \le t}],
\]
and this establishes \eqref{recW}.
\end{proof}
To compute the first moment of the AoI we need to know the second moment
of $T_{\arr(t)}$. Recall that $M_p(t) = \E T_{\arr(t)}^p$ is the $p^{\rm th}$ moment of $T_{\arr(t)}$. These moments can be computed recursively, as in the lemma
below, which is of independent interest.
\begin{lemma}
If $p$ is a positive integer we have
\begin{equation}
M_p(t) = \E M_p(t-\tau) + \E \tau^p + \sum_{k=1}^{p-1} \binom{p}{k}
\E[\tau^k M_{p-k}(t-\tau)] \label{Wpol}
\end{equation}
\end{lemma}
\begin{proof}
Proceed as in the proof of Theorem \ref{Hermes} but let
$f(x)=x^p$ in \eqref{FFF}:
\begin{align*}
M_p(t) &= \E[\tau^p; \tau > t] + \E[(\tau+\Phi_{t-\tau})^p; \tau \le t]
\\
&= \E[\tau^p; \tau > t] 
+ \E\left[\sum_{k=0}^p \binom{p}{k} \tau^k \Phi_{t-\tau}^{p-k}; \tau \le t\right]
\\
&= \E[\tau^p;\tau> t]  + \E[\tau^p; \tau\le t]
+ \E[\Phi_{t-\tau}^{p}; \tau \le t] 
+ \E\left[\sum_{k=1}^{p-1} \binom{p}{k} \tau^k \Phi_{t-\tau}^{p-k}; \tau \le t\right]
\\
&=\E \tau^p +  \E M_p(t-\tau) + \sum_{k=1}^{p-1} \binom{p}{k}
\E[\tau^k M_{p-k}(t-\tau)].
\end{align*}
\end{proof}

Let $(U*U)(t) := \int_0^t U(t-x)\, U(dx)$.
\begin{corollary}
Under the assumptions of Theorem \ref{Hermes},
\begin{equation}
\label{Nero}
\wE \alpha(0) 
= \E \sigma + \frac{\E T_N^2}{2 \E T_N}
= \E \sigma + \frac{\E M_2(\sigma)}{2 \E M_1(\sigma)}
=\E \sigma 
+ \frac{\E \tau^2}{2 \E \tau}
+ \frac{ \E(\tau (U*U)(\sigma-\tau))}{ \E U(\sigma)}.
\end{equation}
\end{corollary}
\begin{proof}
The first equality in \eqref{Nero} follows from the decomposition \eqref{alphadec}.
The second equality follows from
$\E T_N^p = \E \E[T_{\arr(\sigma_0)}^p| \sigma_0] = \E M_p(\sigma)$.
We next have
\begin{equation}
\label{per1}
M_1(t) = \E \tau \, U(t)
\end{equation}
and, from \eqref{Wpol} with $p=2$,
\[
M_2(t) = \E \tau^2 +  2 \E [\tau\,M_1(t-\tau)] + \E M_2(t-\tau)
= \E \tau^2 +  2 \E \tau \, \E [\tau\,U(t-\tau)] + \E M_2(t-\tau).
\]
With the help of \eqref{recU} we can solve this explicitly and express $M_2$
as a function of $U$:
\begin{equation}
\label{per2}
M_2(t) = \E \tau^2 \cdot U(t)+ 2\, \E \tau \, \E[ \tau\, (U*U)(t-\tau)].
\end{equation}
Using \eqref{per1} and \eqref{per2} in the second equality of \eqref{Nero}
we arrive at the third one.
\end{proof}

The Laplace transforms of $U$, $W_{u}$ and $M_2$ are easy to
obtain explicitly in terms of the Laplace transform of $\tau$:
\begin{lemma}
\label{LapLem1}
\begin{align}
&\widehat U(\xi) := \int_0^\infty e^{-\xi t} U(t) dt 
= \frac{1/\xi}{1-\E e^{-\xi \tau}},
\label{hatU}
\\
&\widehat W_{u}(\xi) := \int_0^\infty e^{-\xi t} W_{u}(t) dt
= \frac{1}{\xi} \cdot
\frac{\E[e^{-u\tau} -e^{-(u+\xi) \tau}]}{1-\E e^{-(u+\xi) \tau}}.
\label{hatW}
\\
&\widehat M_2(\xi)  := \int_0^\infty e^{-\xi t} M_2(t) dt 
= \frac{\E \tau^2}{\xi (1-\E e^{-\xi \tau})}
+ 2 (\E \tau) \frac{\E (\tau e^{-\xi\tau})}{\xi (1-\E e^{-\xi \tau})^2}.
\label{hatW2}
\end{align}
\end{lemma}
\begin{proof}
Equation \eqref{recU} then gives
\[
\widehat U(\xi) = \frac{1}{\xi} + \widehat U(\xi)\, \E e^{-\xi \tau},
\]
and hence \eqref{hatU} follows. 
Equation \eqref{recW} gives
\begin{align*}
\widehat W_{u}(\xi) 
&= \int_0^\infty e^{-\xi t}\E[e^{-u\tau} \1_{\tau>t}]\, dt
+ \int_0^\infty e^{-\xi t} \E[W_{u}(t-\tau) e^{-u\tau}  \1_{\tau \le t}]\, dt
\\
&= \E\left[e^{-u\tau}  \frac{1-e^{-\xi \tau}}{\xi}\right]
+ \E\left[  \e^{-u\tau}  e^{-\xi\tau} \int_\tau^\infty e^{-\xi(t-\tau)} W_{u}(t-\tau) dt\right]
\\
&= \frac{1}{\xi} 
\E\left[e^{-u\tau}  (1-e^{-\xi \tau})\right]
+ \E\left[ e^{-u\tau} e^{-\xi\tau} \right] \widehat W_{u}(\xi),
\end{align*}
from which \eqref{hatW} follows. 
Finally, \eqref{hatW2} follows from \eqref{per2} and \eqref{hatU}.
\end{proof}

\begin{corollary} 
\label{cor:AoI-blocking-mean}
Let the assumptions of Theorem \ref{betaiid} hold true.
\\
(i) If the variables $\tau_n$ are exponential with rate $\lambda$, then
\[
\wE e^{-u \alpha(0)} 
= \frac{\lambda}{1+\lambda \E \sigma} \cdot
\frac{(u+\lambda-\lambda \E e^{-u\sigma}) \E e^{-u\sigma}}{u(u+\lambda)},
\qquad
\wE \alpha(0) 
= \E \sigma + \frac{1}{\lambda} + \frac{\lambda}{2}\cdot
\frac{\E \sigma^2}{1+\lambda \E \sigma}.
\]
(ii) If the variables $\sigma_n$ are exponential with rate $\mu$, then
\begin{multline*}
\wE e^{-u \alpha(0)} 
=\frac{1}{\E \tau}\cdot \frac{\mu}{(\mu+u)u} \cdot
\frac{(1-\E e^{-\mu \tau})(1-\E e^{-u \tau})}{1-\E e^{-(\mu+u) \tau}}
= \frac{\mu^2}{(\mu+u)^2} \cdot
\frac{\E e^{-\mu \overline\tau} \, \E e^{-u \overline\tau}}
{\E e^{-(\mu+u) \overline \tau}},
\\
\wE \alpha(0) 
=\frac{1}{\mu} + \frac{\E \tau^2}{2\E \tau} 
+ \frac{\E(\tau e^{-\mu \tau})}{1-\E e^{-\mu \tau}}.
\end{multline*}
(iii) If the $\tau_n$ are  exponential with rate $\lambda$,
and the $\sigma_n$ are exponential with rate $\mu$ then, under $\wP$,
\[
\wE e^{-u \alpha(0)} = \frac{\mu^2\lambda(\lambda+\mu+u)}
{(\lambda+\mu)(\lambda+u)(\mu+u)^2},
\qquad
\wE \alpha(0) = \frac{1}{\mu} + \frac{1}{\lambda} 
+ \frac{\lambda}{\mu(\lambda+\mu)}.
\]
\end{corollary}
\begin{proof}
(ia)
We compute the functions $U(t)$ and $W_{u}(t)$ that enter formula \eqref{alphaBL}.
Since, under $\P$, $\arr=\sum_n \delta_{T_n}$ 
is a Poisson process with a point at $0$
we have, directly from \eqref{funU}, $U(t)=1+\lambda t$.
Since $\tau$ is exponential, \eqref{hatW} explicitly gives the Laplace transform of $W_{u}$:
\[
\widehat W_{u}(\xi) =
 \frac{1}{\xi} \cdot
\frac{\E[e^{-u\tau}(1-e^{-\xi \tau})]}{1-\E e^{-u\tau} e^{-\xi \tau}}
= \frac{1}{\xi} \cdot 
\frac{\frac{\lambda}{\lambda+u} -\frac{\lambda}{\lambda+u+\xi}}
{1-\frac{\lambda}{\lambda+u+\xi}}
= \frac{\lambda}{\lambda+u} \cdot \frac{1}{u+\xi},
\]
and hence
\[
W_{u}(t) = \frac{\lambda}{\lambda+u} e^{-ut}.
\]
Substituting into \eqref{alphaBL} we obtain the announced formula for 
$\wE e^{-u\alpha(0)}$.
\\
(ib) Equations \eqref{hatU} and \eqref{hatW2} give
\[
\widehat M_2(\xi)  
= \E \tau^2\, \widehat U(\xi)
+ 2 (\E \tau) \frac{\E (\tau e^{-\xi\tau})}{\xi (1-\E e^{-\xi \tau})^2}
= \E \tau^2\, \widehat U(\xi) + \frac{ 2 \lambda \E \tau }{\xi^3}.
\]
Hence
\[
M_2(t) =  \E \tau^2\, U(t) + \lambda (\E \tau) t^2.
\]
Using this and $M_1(t) = \E \tau\, U(t)$ in \eqref{Nero} we obtain the announced
formula for $\wE \alpha(0)$.
\\
(iia) 
If $\sigma$ is exponential with rate $\mu$ then
$\E U(\sigma) = \mu \widehat U(\mu)$ and
$\E W_{u}(\sigma) = \mu \widehat W_{u}(\mu)$.
Hence  \eqref{alphaBL} gives
\[
\wE e^{-u \alpha(0)} =\E e^{-u\sigma}\,  \frac{1-\mu \widehat W_{u}(\mu)}
{u\,\E \tau\, \mu\, \widehat U(\mu)}
\]
But the Laplace transforms $\widehat U$ and $\widehat W_{u}$ are known from
Lemma \ref{LapLem1}. Substituting in the last display we obtain the
first announced equality
for $\wE e^{-u \alpha(0)}$.
For the second equality, simply replace the three terms of the form $1-\E e^{-\xi \tau}$ by 
$\xi (\E \tau) \E e^{-\xi\overline\tau}$.
\\
(iib) From the middle of \eqref{Nero} we have
\[
\wE \alpha(0) = \frac{1}{\mu} + \frac{\widehat W_2(\mu)}{2 \widehat W_1(\mu)}
\]
and the formula follows from the previously derived formulas for $\widehat W_2$
and $\widehat W_1$.
\\
(iiia)
Consider the second equality in (ii).
Since $\overline \tau \eqdist \tau$ we have $\E e^{-\xi \overline \tau}
= \lambda/(\xi+\lambda)$.
Replacing the three terms in the second equality in (ii) by such ratios
we arrive at the announced formula.
Alternatively, letting $\E e^{-\mu \sigma} = \mu/(\mu+u)$ and $\E \sigma=1/\mu$
in (i) we arrive at the same formula.
\\
(iiib)
Set $\E \sigma = 1/\mu$, $\E \sigma^2 = 2/\mu^2$ in the last formula of (i).
\end{proof}

\subsection{The new age of information  for the blocking system}
\label{naoibl}
Recall that the NAoI process is given by $\beta(t) = A_t - S_{D_t}$,
where $A_t$ is the last arrival (accepted or not) before $t$
and $S_{D_t}$ is the last successful arrival before the last successful departure before $t$;
this quantity is given by \eqref{onetwo}.

\begin{theorem}
\label{betagen2}
Consider the blocking system under stationarity assumptions
and assume that \eqref{SenecaB} holds.
Then the $\wP$-law of $\beta(0)$ has an atom at $0$ satisfying
\begin{equation}
\label{0value}
\wP(\beta(0)=0)
= \frac{\Estar(\tau_0-\sigma_0)^+}{\Estar T_N},
\end{equation}
while, for $f$ bounded and measurable function,
\begin{multline}
\label{posvalue}
\wE [f(\beta(0)); \beta(0)>0] 
= \frac{1}{\Estar T_N} \Estar\left\{ \sum_{i=0}^{N-1} \tau_i f(T_i - T_M) 
- (T_N-\sigma_0)\,  f(T_{N-1} - T_M) \right\}
\\
+ \frac{1}{\Estar T_N} \Estar\left\{(T_N-\sigma_0)\,  f(T_{N-1})\1_{T_{N-1}>0}\right\},
\end{multline}
where $N=\inf\{\ell \ge 1:\, \psi_\ell=1\}$ and
$M=\sup\{\ell \le -1:\, \psi_\ell=1\}$.
\end{theorem}

\begin{proof}
Notice that $N = \inf\{\ell \ge 1:\, T_\ell  \ge \sigma_0\}$, $\P^*$-a.s.
We use the Palm inversion formula:
\begin{equation}
\label{palminvagain}
\wE f(\beta(0)) = \frac{\Estar \int_{B_0}^{B_1} f(\beta(t))\, dt}
{\Estar(B_1-B_0)}.
\end{equation}
Since $M, N$ are the  indices of the admitted messages nearest to $0$,
\[
B_{-1} = T_M \le T_{-1} < T_0=B_0=0 < T_1 < \cdots < T_{N-1} < \sigma_0 \le T_N=B_1,
\quad \Pstar\text{-a.s.}
\]
In particular, $B_1-B_0=T_N$, $\P^*$-a.s.
Since $\beta(t) = A_t-S_{D_t}$, using \eqref{onetwo} we have
\[
\beta(t) = \begin{cases}
T_i-T_M, &\text{ if } T_0 \le T_i \le t < T_{i+1} \le T_{N-1} 
\\
T_{N-1}-T_M, &\text{ if } T_{N-1} \le t < T_0+\sigma_0
\\
T_{N-1}-T_0, &\text{ if }T_0+\sigma_0 \le t < T_N
\end{cases}.
\]
Let $f: \R \to \R$ be bounded and measurable.
We write the integral in the numerator of \eqref{palminvagain} as:
\begin{align}
\int_{T_0}^{T_N} f(\beta(t)) \, dt
&= \int_{T_0}^{T_{N-1}} f(\beta(t))\, dt
+\int_{T_{N-1}}^{T_0+\sigma_0} f(\beta(t))\, dt
+\int_{T_0+\sigma_0}^{T_{N}} f(\beta(t))\, dt
\nonumber
\\
&= \sum_{i=0}^{N-2}  \int_{T_i}^{T_{i+1}} f(T_i - T_M) dt
+ \int_{T_{N-1}}^{T_0+\sigma_0} f(T_{N-1} - T_M) dt
+ \int_{T_0+\sigma_0}^{T_N} f(T_{N-1} - T_0) dt
\nonumber
\\
&{= \sum_{i=0}^{N-2} \tau_i f(T_i - T_M) 
+ (\sigma_0-T_{N-1})\,  f(T_{N-1} - T_M) 
+ (T_N-\sigma_0)\, f(T_{N-1})}.
\label{betaintegral}
\end{align}
Add and subtract the term corresponding to $i=N-1$ to write the last line as
\begin{align}
&= 
\sum_{i=0}^{N-1} \tau_i f(T_i - T_M) - \tau_{N-1} f(T_{N-1} - T_M)
+ (\sigma_0-T_{N-1})\,  f(T_{N-1} - T_M) 
+ (T_N-\sigma_0)\, f(T_{N-1})  \nonumber
\\
&= \sum_{i=0}^{N-1} \tau_i f(T_i - T_M) 
+ (\sigma_0-T_{N-1}-\tau_{N-1})\,  f(T_{N-1} - T_M) 
+ (T_N-\sigma_0)\, f(T_{N-1}) \nonumber
\\
&= \left\{ \sum_{i=0}^{N-1} \tau_i f(T_i - T_M) 
- (T_N-\sigma_0)\,  f(T_{N-1} - T_M) \right\}
+ (T_N-\sigma_0)\, f(T_{N-1}).
\label{bracketed}
\end{align}
(For $f \ge 0$, the term in the bracket is positive because the last term of the sum
is $\tau_{N-1} f(T_{N-1}-T_M)$ is bigger than 
$(T_N-\sigma_0)\,  f(T_{N-1} - T_M)$ and this
is because
$\tau_{N-1}-(T_N-\sigma_0) = \sigma_0 - T_{N-1} >0$.)
By \eqref{palminvagain},
\begin{equation}
\label{straight}
\Estar T_N \, \wE f(\beta(0)) 
= 
 \Estar\left\{ \sum_{i=0}^{N-1} \tau_i f(T_i - T_M) 
- (T_N-\sigma_0)\,  f(T_{N-1} - T_M) \right\}
+\Estar\left\{(T_N-\sigma_0)\,  f(T_{N-1})\right\}.
\end{equation}
To reveal the atom of the $\wP$-law of $\beta(0)$ at $0$, let
\[
f(x) = \1_{x=0}.
\]
Then $f(T_i - T_M) =0$ because $T_M<0$.
Also, $f(T_{N-1})=\1_{T_{N-1}=0} = \1_{N=1} = \1_{\tau_0 \ge \sigma_0}$.
Hence
\[
\Estar T_N \, \wP(\beta(0)=0)
= \Estar\left\{(T_N-\sigma_0)\, \1_{N=1}\right\}
= \Estar\left\{(\tau_0-\sigma_0)\, \1_{\tau_0 \ge \sigma_0}\right\}
= \Estar(\tau_0-\sigma_0)^+.
\]
On the other hand,
\begin{align*}
\wE [f(\beta(0)); \beta(0)>0]
&=  \wE f(\beta(0)) - \wE [f(\beta(0)); \beta(0)=0]
\\
&=  \wE f(\beta(0)) - f(0) \wP(\beta(0)=0)
\\
&=  \wE [f(\beta(0)) - f(0) \1_{\beta(0)=0}] \equiv \wE g(\beta(0)),
\end{align*}
where
\[
g(x) = f(x) - f(0) \1_{x=0}.
\]
We use $g$ in place of $f$ in \eqref{straight} after noting that
$g(T_i-T_M) = f(T_i-T_M) - f(0) \1(T_i=T_M) = f(T_i-T_M)$ for $i \ge 0$,
and $g(T_{N-1}) = f(T_{N-1}) - f(0) \1_{T_{N-1}=0} = f(T_{N-1}) - f(0) \1_{N=1}$.
So
\begin{multline*}
\Estar T_N \, \wE [f(\beta(0)) ; \beta(0)>0]
= 
\\
= \Estar\left\{ \sum_{i=0}^{N-1} \tau_i f(T_i - T_M) 
- (T_N-\sigma_0)\,  f(T_{N-1} - T_M) \right\}
+\Estar\left\{(T_N-\sigma_0)\,  (f(T_{N-1}) - 
f(0) \1_{N=1})\right\}.
\end{multline*}
Notice that 
\[
f(T_{N-1}) - f(0) \1_{N=1}
= f(T_{N-1}) - f(T_{N-1}) \1_{N=1}
= f(T_{N-1}) \, \1_{N>1}
= f(T_{N-1}) \, \1_{T_{N-1}>0}
\]
and substitute into the last display to obtain the announced formula.
\end{proof}

By Palm theory and stationarity, we have that $|M|$ and $N$ have the same $\Pstar$-law
and so do $|T_M|$ and $T_N$.
This  simple fact is stated as an stand-alone lemma because
it holds only under stationary assumptions and because it is needed when we
explicitly compute distributions under independence assumptions.

\begin{remark}
\label{rembet} 
We now give a physical meaning to the $\wP$-law of $\beta(0)$
conditional on $\beta(0)>0$.
Say that the message arriving
at time $T_n$ is {\it undisturbed} if it is admitted (and hence successful) and
no other messages arrive during the time it is being processed;
i.e., $\psi_n=1$ and $T_n + \sigma_n \le T_{n+1}$.
Therefore, for $T_n+\sigma_n \le t < T_{n+1}$ we have $\beta(t)=0$:
undisturbed messages provided the freshest possible information;
this is what contributes to the atom at $0$ for $\beta(0)$.
Define then an auxiliary system, pathwise, by removing all undisturbed messages.
If $\beta_+(t)$ denotes the NAoI process for the auxiliary system then
we have that, under $\wP$, $\beta(0)$ equals $0$ with probability
$\frac{\Estar(\tau_0-\sigma_0)^+}{\Estar T_N}$
or $\beta_+(0)$ with the remaining probability.
In particular, 
\[
\wE [f(\beta(0));\, \beta(0)>0] = \wE f(\beta_+(0)).
\]
\end{remark}

\begin{lemma}
\label{lemtriv}
Assume that $(\tau_n, \sigma_n)$, $n \in \Z$, is stationary under $\P$.
Let $N = \inf\{\ell \ge 1:\, \psi_\ell=1\}$
and $M= \sup\{\ell \le -1:\, \psi_\ell=1\}$.
Then
\[
\E (g(-T_M)|\psi_0=1) = \E (g(T_N)|\psi_0=1),
\]
for any bounded and measurable function $g$.
\end{lemma}
\begin{proof}
The point process $\sum_n \psi_n \delta_n$ is stationary under $\widetilde \P$
and the Palm probability of the latter with respect to this point process
is denoted by $\Pstar$.
If $\cdots < T^*_{-1} < T^*_0 \le 0 < T^*_1 < T^*_2<\cdots$ is
an enumeration of the points of $\sum_n \psi_n \delta_n$ in their natural order
then $\Estar g(-T^*_{-1}) = \Estar g(T^*_1)$ for any bounded measurable function $g$.
But $T^*_1=T_N$ and $T^*_{-1} = T_M$ and $\P=\Pstar(\cdot|\psi_0=1)$.
\end{proof}

Under i.i.d.\ assumptions, and because the decision
on whether to admit a message or not is past-dependent, the ensued regeneration
results into further simplification and the vanishing of the $M$ from the formula.
We explain this below.
First fix $u \ge 0$ and consider the function $W_u(t)$ introduced in \eqref{funW}
as well as
{
\begin{align}
V_u(t) &:= \E \sum_{i=0}^{\arr(t)-1}  e^{-u T_i}, \quad t \ge 0  \label{funV}
\\
Q_u(t) &:= \E  \big\{ (T_{\arr(t)}-t)\, e^{-u T_{\arr(t)-1}} \big\}, \quad t \ge 0.  \label{funQ}
\end{align}
}

\begin{theorem}
\label{Zeus}
Consider the blocking system and
assume that $(\tau_n, \sigma_n)$, $n \in \Z$, is i.i.d.\ under $\P$
and such that $\E \tau_0 <\infty$ and $\P(\tau_0 \ge \sigma_0)>0$.
Assume further that $\tau_n$ is independent of $\sigma_n$ for all $n$.
Then $\wP(\beta(0)=0) = \frac{\E(\tau_0-\sigma_0)^+}{\E T_N}$ and
{
\begin{align}
\wE [e^{-u \beta(0)}; \beta(0)>0] 
&
=  \frac{1}{\E T_N}
\E e^{-u T_N}\,
\left\{\E \tau\, \E \sum_{i=0}^{N-1}  e^{-u T_i} - \E (T_N-\sigma_0)\,  e^{-u T_{N-1}}
\right\}
\nonumber
\\
& \qquad\qquad\qquad\qquad\qquad\qquad\qquad+  \frac{1}{\E T_N}
\E\left\{ (T_N-\sigma_0)\, e^{-u T_{N-1}} \1_{T_{N-1}>0}\right\}
\nonumber
\\
& =\frac{\E W_{u}(\sigma) \,
[ \E \tau\, \E V_u(\sigma) - \E Q_u(\sigma)] + \E [e^{-u\tau} Q_u(\sigma-\tau)]}{\E \tau\, \E U(\sigma)},
\label{betaBL}
\end{align}
}
where $U, W_u$ are unique solutions to the fixed point equations
\eqref{recU}, \eqref{recW}, respectively, while $V_u, Q_u$ 
are unique solutions to
{
\begin{align}
V_u(t) 
&= 1+ \int_{(0,t]} V_u(t-x)\, e^{-ux}\, \P(\tau\in dx),
\label{recV}
\\
Q_u(t) 
&= \E (\tau-t)^+ + \int_{(0,t]} Q_u(t-x) e^{-ux} \P(\tau\in  dx).
\label{recQ}
\end{align}}
In particular, under $\wP$, and conditional on $\beta(0)>0$, 
the random variable $\beta(0)$ is absolutely continuous.
\end{theorem}

\begin{remark}
The term 
$\E \tau\, \E V_u(\sigma) - \E Q_u(\sigma)$
in \eqref{betaBL} is nonnegative and this is due to the  remark made below
\eqref{bracketed} about the nonnegativity of the bracketed term in \eqref{betaBL}.
\end{remark}

\begin{proof}
The value of $\wP(\beta(0)=0)$ follows from \eqref{0value} and \eqref{psim} that
allows us to replace $\Estar$ by $\E$.
To show the rest, we look at the various terms in \eqref{posvalue} with $f(x) = e^{-ux}$.
Using 
\eqref{psim} we obtain
\begin{equation}
\label{plato1}
\Estar \sum_{i=0}^{N-1} \tau_i e^{-u (T_i - T_M) }
=\E (e^{ u  T_M} |\psi_0=1)\, \E \sum_{i=0}^{N-1} \tau_i e^{-u T_i}
\end{equation}
Due to Lemma \ref{lemtriv}, the first term of the product is further written as:
\[
\E (e^{u T_M}  |\psi_0=1) 
= \E(e^{-u T_N} | \psi_0=1) 
= \E e^{-u T_N} .
\]
The second term in the last product of \eqref{plato1} is computed as follows.
\begin{multline}
 \E \sum_{i=0}^{N-1} \tau_i e^{-u T_i}
= \E \sum_{i=0}^\infty \tau_i e^{-u T_i} \1_{T_i < \sigma_0}
= \sum_{i=0}^\infty \E\{\E[ \tau_i e^{-u T_i} \1_{T_i < \sigma_0}|
\sigma_0, \tau_0,\ldots, \tau_{i-1}]\}
\\
= \sum_{i=0}^\infty \E\{e^{-u T_i} \1_{T_i < \sigma_0}\E[ \tau_i |
\sigma_0, \tau_0,\ldots, \tau_{i-1}]\}
= (\E \tau)\, \E \sum_{i=0}^{N-1} e^{-u T_i}.
\label{plato2}
\end{multline}
Using the same logic,
\begin{gather}
\Estar (T_N-\sigma_0) e^{-u (T_{N-1}-T_M)}
= \Estar e^{ uT_M} \, \E(T_N-\sigma_0)\,  e^{-u T_{N-1}}
=  \E e^{ uT_M} \, \E(T_N-\sigma_0)\,  e^{-u T_{N-1}}
\label{plato3}
\\
\Estar\left\{(T_N-\sigma_0)\,  f(T_{N-1})\1_{T_{N-1}>0}\right\}
= \E \left\{(T_N-\sigma_0)\,  f(T_{N-1})\1_{T_{N-1}>0}\right\}
\label{plato4}
\end{gather} 
Substituting \eqref{plato2} into \eqref{plato1} and then this, together with
\eqref{plato3} and \eqref{plato4}, into \eqref{posvalue} we arrive at the first 
equality for \eqref{betaBL}.
For the second equality, use \eqref{ETN}, \eqref{EuTN} and
\eqref{funW}, \eqref{funV}, \eqref{funQ} and 
observe that
\[
\E  \big\{ (T_{\arr(t)}-t)\, e^{-u T_{\arr(t)-1}} \1_{T_{\arr(t)}-1>0} \big\}
= \E[ e^{-u\tau}\, Q_u(t-\tau)].
\]
\\
To see that $V_u$ satisfies \eqref{recV}, notice that
\begin{align*}
V_u(t) &= \E \left[\sum_{i=0}^{\arr(t)-1}  e^{-u T_i} ;\, T_1 > t\right]
+ \E \left[\sum_{i=0}^{\arr(t)-1}  e^{-u T_i} ;\, T_1 \le  t\right]
\\
&=\E[e^{-u T_0};\,T_1>t] 
+ \E \left[e^{-u T_0}+ e^{-u T_1} V_u(t-T_1);\, t-T_1 \ge 0\right]
\\
&= e^{-u 0} + \E[e^{-u \tau} V_u(t-\tau) \1_{\tau \le t}].
\end{align*}
To see that $Q_u$ satisfies \eqref{recQ}, notice that
\begin{align*}
Q_u(t) &= \E \left[  (T_{\arr(t)}-t)\, e^{-u T_{\arr(t)-1}};\, T_1 > t\right]
+ \E \left[ (T_{\arr(t)}-t)\, e^{-u T_{\arr(t)-1}};\, T_1 \le  t\right]
\\
&=\E \left[  (T_{1}-t)\, e^{-u T_{0}};\, T_1 > t\right]
+\int \E \left[  (x+T_{\arr(t-x)}-t)\, e^{-u (x+T_{\arr(t-x)-1})}\right]\,\1_{x \le t}\, \P(T_1 \in dx)
\\
&= \E \left[  (\tau-t)\, e^{-u 0};\, \tau > t\right]
+\int_{(0,t]} e^{-u x}\,\E \left[  (T_{\arr(t-x)}-(t-x))\, e^{-u (T_{\arr(t-x)-1})}\right]\, 
\P(\tau \in dx)
\\
&=\E(\tau-t)^+
+\int_{(0,t]} e^{-u x}\,Q_u(t-x)\, 
\P(\tau \in dx).
\end{align*}
\end{proof}

Continuing in the same manner as Lemma \ref{LapLem1}, we obtain the
Laplace transforms of $V_u$ and $Q_u$.
\begin{lemma}
\label{LapLem2}
\begin{align}
&\widehat V_u(\xi) = \frac{1/\xi}{1-\E e^{-(u+\xi)\tau}}
\label{hatV}
\\
&\widehat Q_u(\xi) = \frac{1}{\xi^2} \frac{\xi \E \tau-1 + \E e^{-\xi \tau}}{1-\E e^{-(u+\xi)\tau}}
\label{hatQ}
\end{align}
\end{lemma}
\begin{proof}
Directly from \eqref{recV} and \eqref{recQ}.
\end{proof}

\begin{corollary}
\label{cor:NAoI-blocking-mean}
Let the assumptions of Theorem \ref{betaiid} hold true.
\\
(i) If the variables $\tau_n$ are exponential with rate $\lambda$, then
\[
\wP(\beta(0)=0) = \frac{\E e^{-\lambda \sigma}}{1+\lambda \E\sigma}
\]
and, with $L_\sigma(u)=\E e^{-u\sigma}$,
\[
\wE[e^{-u\beta(0)}; \beta(0)>0]
=
\frac{\lambda}{1+\lambda \E \sigma}
\bigg[
\frac{L_\sigma(u)}{\lambda+u} \,
\frac{\lambda^2(1-L_\sigma(u))-u^2(1-L_\sigma(\lambda))}{u(\lambda-u)} 
+ \frac{L_\sigma(u)-L_\sigma(\lambda)}{\lambda-u}
\bigg]
\]
(ii) If the variables $\sigma_n$ are exponential with rate $\mu$, then,
with $L_\tau(u) = \E e^{-u \tau}$,
\[
\wP(\beta(0)=0) =
\frac{1}{\mu \E \tau}
(1-\E e^{-\mu \tau}) (\mu \E \tau - 1 - \E e^{-\mu \tau}),
\]
\[
\wE[e^{-u\beta(0)}; \beta(0)>0] 
= \frac{1-L_\tau(\mu)}{\mu \E \tau (1-L_\tau(u+\mu))}
\bigg[
\frac{L_\tau(u)-L_\tau(u+\mu)}{1-L_\tau(u+\mu)} (1-L_\tau(\mu))
+ L_\tau(u+\mu) \, (\mu \E \tau - 1 - L_\tau(\mu))
\bigg]
\]
(iii) If the $\tau_n$ are  exponential with rate $\lambda$,
and the $\sigma_n$ are exponential with rate $\mu$ then, under $\wP$,
\[
\beta(0) \eqdist
\begin{cases}
0 , & \text{with probability } \frac{\mu^2}{(\lambda+\mu)^2}
\\[4mm]
\zeta, & \text{with probability } \frac{\lambda(\lambda+2\mu)}{(\lambda+\mu)^2}
\end{cases},
\]
where $\zeta$ is an absolutely continuous random variable with
\[
\E e^{-u \zeta} 
= \frac{\mu^2}{\lambda+2\mu}\,
\frac{u^2 + (2\lambda+\mu)u + \lambda(\lambda+2\mu)}{(u+\lambda)(u+\mu)^2}.
\]
\end{corollary}

\begin{proof} 
From Theorem  \eqref{Zeus}, we
have $\wP(\beta(0)=0) = \lambda \E(\tau-\sigma)^+/ \E U(\sigma)$
and the expressions of this are obtained by elementary integrals in all
cases.
We rewrite \eqref{betaBL} as
\begin{equation}
\wE [e^{-u \beta(0)}; \beta(0)>0] 
=\frac{\E W_{u}(\sigma) \,
\E H_u(\sigma) + \E Q_u^+(\sigma)}{\E \tau\, \E U(\sigma)},
\label{betaBL2}
\end{equation}
where
\[
H_u(t) = \E \tau\,  V_u(t) - Q_u(t),\quad
Q_u^+(t) = \E[ e^{-u\tau}\, Q_u(t-\tau)].
\]
We thus know the Laplace transforms of all functions entering in \eqref{betaBL2}
in terms of $L_\tau(\xi):= \E e^{-\xi \tau}$:
\begin{multline*}
\widehat U(\xi)  = \frac{1/\xi}{1-L_\tau(\xi)},\quad
\widehat W_{u}(\xi)  
= \frac{1}{\xi} \frac{L_\tau(u) - L_\tau(u+\xi)}{1-L_\tau(u+\xi)},\quad
\\
\widehat H_u(\xi) 
= \frac{1}{\xi^2}\frac{1- L_\tau(\xi)}{1-L_\tau(u+\xi)},\quad
\widehat Q_u^+(\xi) = \frac{L_\tau(u+\xi)}{\xi^2} \,
\frac{\xi \E \tau - 1 + L_\tau(\xi)}{1-L_\tau(u+\xi)}.
\end{multline*}
(i) When $\tau$ is exponential, we already know that $U(t)=1+\lambda t$ and that $W_u(t) = \lambda e^{-ut}/(\lambda+u)$ and, with $L_\tau(u)=\lambda/(\lambda+u)$,
we obtain
\[
\widehat H_u(\xi)  = \frac{\lambda+u+\xi}{\xi (\lambda+\xi) (u+\xi)},
\quad
\widehat Q_u^+(\xi)  = \frac{1}{(u+\xi)(\lambda+\xi)},
\]
that can easily be inverted to the nonnegative functions
\[
H_u(t) 
= \frac{\lambda^2 (1-e^{-ut}) - u^2 (1-e^{-\lambda t})}{\lambda u (\lambda-u)},
\quad
Q_u^+(t) =\frac{ e^{-ut} - e^{-\lambda t}}{\lambda-u} .
\]
The values of $H_u$ and $Q_u^+$ at $u=\lambda$ should be
interpreted as limits when $u \to \lambda$. Thus,
$H_\lambda(t) = \lambda^{-1} [2-(\lambda t+2)] e^{-\lambda t}$,
$Q_\lambda^+(t) =t e^{-\lambda t}$.
Substitute these functions in \eqref{betaBL2} to obtain the announced formula.
\\
(ii) When $\sigma$ is exponential with rate $\mu$, all functions in \eqref{betaBL2}
are essentially Laplace transforms of $\sigma$, for example,
$\E W_u(\sigma) = \mu \widehat W_u(\mu)$. Hence
\begin{equation*}
\wE [e^{-u \beta(0)}; \beta(0)>0] 
=\frac{\mu \widehat W_u(\mu) \,
\mu \widehat H_u(\mu)+ \mu \widehat Q_u^+(\mu)}{\E \tau\, \mu \widehat U(\mu)},
\end{equation*}
and the formula is obtained because we know all Laplace transforms.
\\
(iii) 
The formula readily follows from either (i) or (ii).
\end{proof}

Let us take a closer look at the law of the random variable $\zeta$
of Corollary \ref{cor:NAoI-blocking-mean}(iii).
Letting $\rho=\lambda\mu$ we have
\[
\E e^{-u \mu \zeta} 
= \frac{1}{\rho+2}\,
\frac{u^2 + (2\rho+1)u + \rho(\rho+2)}{(u+\rho)(u+1)^2}.
\]
Inverting this Laplace transform, we find that $\mu \zeta$ has density
\[
g_\rho(t) 
= \frac{1}{(\rho+2)(\rho-1)^2}
\big[ \rho e^{-\rho t} + (\rho^3-3\rho+1+\rho^2(\rho-1)t) e^{-t}
\big],
\]
for all values of $\rho\neq 1$ and, for $\rho=1$, the density corresponds
to the limit of this expression when $\rho \to 1$:
\[
g_1(t) = \frac16 (t^2+2t+2) e^{-t}.
\]

We now pass on to computing first moments.
\begin{lemma}
Consider the blocking system under stationarity assumptions.
Then 
\begin{equation}
\label{bebe}
\wE \beta(0)=
\frac{\displaystyle \Estar\left[\sum_{i=0}^{N-1} \tau_i T_i - \sigma_0 T_M\right]}
{\displaystyle \Estar T_N}\,
\end{equation}
\end{lemma}
\begin{proof}
Take $f(x)=x$ in \eqref{betaintegral} and regroup the terms there to obtain
\begin{align*}
\int_{T_0}^{T_N} \beta(t) dt 
=
\sum_{i=0}^{N-1} \tau_i (T_i - T_M) 
+ (T_N-\sigma_0)\,T_M 
= \sum_{i=0}^{N-1} \tau_i T_i - \sigma_0 T_M
\end{align*}
and then use the Palm inversion formula.
\end{proof}

Next define
\begin{equation}
\label{funZ}
Z(t) = \E \sum_{i=0}^{\arr(t)-1} T_i, \quad t \ge 0.
\end{equation}

\begin{lemma}
Consider the blocking system and
assume that $(\tau_n, \sigma_n)$, $n \in \Z$, is i.i.d.\ under $\P$
and such that $\E \tau_0 <\infty$.
Assume further that $\tau_n$ is independent of $\sigma_n$ for all $n$.
Then
\[
\wE \beta(0)= \E \sigma + \frac{\E \left[\sum_{i=0}^{N-1} T_i\right] }{\E N}
= \E \sigma + \frac{\E Z(\sigma)}{\E U(\sigma)},
\]
where $Z$ is the unique solution to the fixed-point equation
\begin{gather*}
Z(t) =  \E[Z(t-\tau)] + \E[\tau U(t-\tau)]
\end{gather*}
and has Laplace transform
\[
\widehat Z(\xi) = \frac{\E\tau e^{-\xi \tau}}{\xi(1-\E e^{-\xi \tau})^2}.
\]
\end{lemma}
\begin{proof}
The numerator of \eqref{bebe} is written as
\begin{align*}
\Estar\left[\sum_{i=0}^{N-1} \tau_i T_i -\sigma_0 T_M \right]
&=\Estar\sum_{i=0}^{N-1} \tau_i T_i + \Estar \sigma_0 (-T_M)
\\
&= \E \tau\, \E \left[\sum_{i=0}^{N-1} T_i\right]  + \E(- T_M|\psi_0=1) \, \E \sigma
\\
&= \E \tau\, \E \left[\sum_{i=0}^{N-1} T_i\right]  + \E T_N \, \E \sigma.
\end{align*}
Dividing this by $\E T_N = \E \tau\, \E N$ results in the first equality.
Next use the function \eqref{funZ} to write
$\E \left[\sum_{i=0}^{N-1} T_i\right] = \E Z(\sigma)$.
The fixed point equation is obtained from first principles or by differentiating
both sides of \eqref{recV} with respect to $u$ and letting $u \to 0$.
The Laplace transform is obtained by taking the Laplace transform of both sides
of the fixed-point equation.
\end{proof}

\begin{corollary}
Let the assumptions of Theorem \ref{betaiid} hold true.
\\
(i) If the variables $\tau_n$ are exponential with rate $\lambda$, then
\[
\wE \beta(0)= \E \sigma + \frac{\lambda}{2}\frac{ \E \sigma^2}
{1+\lambda \E \sigma}.
\]
(ii) If the variables $\sigma_n$ are exponential with rate $\mu$, then,
with $L_u = \E e^{-u \tau}$,
\[
\wE \beta(0)
=
\frac{1}{\mu} + \frac{\E\tau e^{-\mu \tau}}{1-\E e^{-\mu \tau}}
\]
(iii) If the $\tau_n$ are  exponential with rate $\lambda$,
and the $\sigma_n$ are exponential with rate $\mu$ then, under $\wP$,
\[
\wE \beta(0)
= \frac{1}{\mu} +\frac{\lambda}{\mu(\lambda+\mu)}.
\]
\end{corollary}

\section{Concluding discussion and open problems} \label{sec:fw}

\paragraph{Summary.}
In summary, the contributions of this paper are:
A new age of information measure (NAoI) definition was introduced and motivated.
The utility of Palm calculus was demonstrated in deriving the distribution of 
AoI and NAoI for stationary bufferless systems under pushout and blocking policies.
All formulas obtained for bufferless $\PO$ are also valid for inifinite buffer queues
under preemptive LIFO policy; see Section \ref{bb}.
In particular, the expectations of these quantities, under renewal assumptions
are summarized in Table \ref{tab:means}. 
Under the same assumptions, some interesting stochastic decomposition and
representation results were also obtained; see Section \ref{sec:outline} for a summary
of these results.

\begin{table}[h]
\renewcommand{\arraystretch}{1.5}
\centering
\begin{tabular}{ |c|c|c|c|c| }
\hline
$~$ & \multicolumn{2}{ |c| }{pushout ($\PO$)}  & \multicolumn{2}{ |c| }{blocking
($\BL$)} \\ \cline{2-5} 
model  & AoI ($\wE \alpha_\PO(0)$) & NAoI ($\wE \beta_\PO(0)$)   & AoI ($\wE \alpha_\BL(0)$)  & NAoI ($\wE \beta_\BL(0)$) \\ \hline\hline

GI/GI & $\frac{\E \tau^2}{2\E \tau} + \frac{\E \tau \wedge \sigma}{\P (\tau\geq\sigma)}$ 
& $ \frac{\E \tau \wedge \sigma}{\P (\tau\geq\sigma)}$  
& $\frac{1}{\mu} + \frac{\E \tau^2}{2\E \tau} 
+ \frac{\E \tau (U*U)(\sigma-\tau)}{\E U(\tau)}$ 
& $ \frac{1}{\mu} + \frac{\E Z(\sigma)}{E U(\sigma)}$\\ [2pt]\hline

M/GI &  $\frac{1}{\lambda \E e^{-\lambda\sigma}} $
&  $\frac{1}{\lambda \E e^{-\lambda\sigma}} - \frac{1}{\lambda}$
& $\frac{1}{\mu} + \frac{1}{\lambda}+\frac{\lambda}{2}
\frac{\E \sigma^2}{1+\lambda\E \sigma} $
&  $\frac{1}{\mu} + \frac{\lambda}{2}\frac{\E\sigma^2}{1+\lambda/\mu}$\\ [2pt]\hline

GI/M & $\frac{\E \tau^2}{2\E \tau} + \frac{1}{\mu}$
& $\frac{1}{\mu}$
& $\frac{1}{\mu}+ \frac{\E \tau^2}{2\E \tau} 
+\frac{\E \tau e^{-\mu \tau}}{1-\E e^{-\mu\tau}}$
& $\frac{1}{\mu} + \frac{\E \tau e^{-\mu\tau}}{1-\E e^{-\mu\tau}}$\\ [2pt]\hline

M/M &  $\frac{1}{\lambda}+\frac{1}{\mu} $
& $\frac{1}{\mu}$
& $\frac{1}{\mu} + \frac{1}{\lambda}+
\frac{\lambda}{\mu(\lambda+\mu)}$
& $\frac{1}{\mu} + \frac{\lambda}{\mu(\lambda+\mu)}$\\[2pt]\hline
\end{tabular}
\caption{Mean AoI and NAoI for different models of
interarrival times (with $\E \tau=1/\lambda$) and service 
times (with $\E \sigma = 1/\mu$) in the renewal case. 
}\label{tab:means}
\end{table}

Using Laplace inversion, we obtained, in certain cases, 
the density of AoI and the density of the NAoI conditional that it be positive.
We may alternately obtain expressions for the probability densities
by  using {\em level-crossing arguments} as in, e.g., \cite{Brill08}.
We should also point out the generality of the formulas obtained
in Theorems \ref{alphagen}, \ref{betagen}, \ref{alphagen2} and \ref{betagen2}:
they remain true even under general stationarity assumptions.
Therefore, we can, for example, 
incorporate situations where messages arrive according to
processes that are more general than renewal ones, e.g., Markov renewal.

\paragraph{What is best for a bufferless system?}
Let us now take a look at the issue of choosing the ``best'' policy
for bufferless system. The choice depends not only on the arrival/processing 
rates but on the way that arrivals and processing times are distributed.
It also depends on what we mean by ``best''.
If ``good'' means low expectation and if renewal assumptions are made,
then sometimes $\PO$ always outperforms $\BL$, sometimes
$\BL$ outperforms $\PO$ and sometimes the answer depends on how loaded the
system is. Suppose $\mu=1$.\\
(a) In the M/M case we have, for all $\lambda$,
\[
\widetilde \E \beta_\PO=1 <
\widetilde\E \beta_\BL=1+\frac{\lambda}{\lambda+1}.
\]
(b) In the M/D case (where D stands for deterministic) we have,  for all $\lambda$,
\[
\widetilde\E \beta_\PO = \frac{e^{\lambda}-1}{\lambda}
> \widetilde\E \beta_\BL =1 + \frac{1}{2} \cdot \frac{\lambda}{1+\lambda}.
\]
This inequality is implied by the inequality
$e^x > 1+x+x^2/2$ which is true for all $x>0$.
\\
(c) In the M/GI case we have a freedom to choose the law of $\sigma$.
We take a mixture: $\sigma$ is either equal to the constant $1/3$,
or is exponentially distributed with rate $3/5$, with equal probability for each case
(the parameters are chosen so that $\E \sigma=1$) we have that $\PO$ outperforms
$\BL$ for high arrival rates (roughly for $\lambda >11.2$)
but the opposite is true for smaller rates. The exact expressions are obtained
from the second row of Table \ref{tab:means} and are plotted in 
Figure \ref{fig3}(c).
\begin{figure}[t]
    \centering
    \subfloat[$M/M$]{
        \includegraphics[width=0.3\textwidth,height=2.8cm]{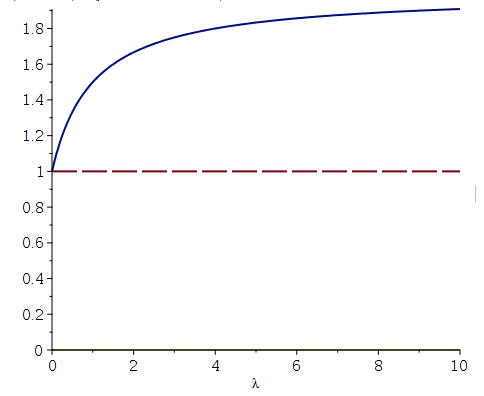}
        }
    \hfill
    \subfloat[$M/D$]{
        \includegraphics[width=0.3\textwidth,height=4cm]{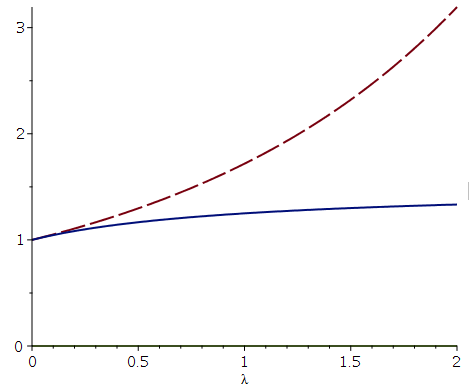}
        }
\hfill
    \subfloat[$M/GI$]{
        \includegraphics[width=0.3\textwidth,height=4.1cm]{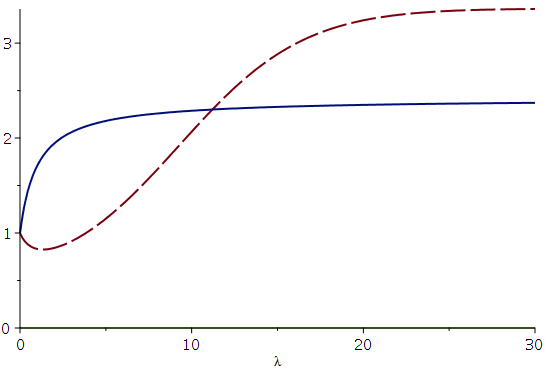}
        }
    \caption{\it 
Mean NAoI as a function of the arrival rate $\lambda$
for pushout (broken line) and blocking (solid line) policies
in three cases.
Here, $\mu=1$ in all cases.}
\label{fig3}
\end{figure}
A better policy, insofar as expectations are concerned, 
can be found by considering a $\PO\BL(\ell)$ policy
or a $\BL\PO(\ell)$ policy for appropriate $\ell$.
Changing the optimality criterion changes the story completely.

\paragraph{Large buffers make no sense.}
Most of research in the AoI area so far has focused on systems with 
infinite storage capacity have been studied. Among these systems, FIFO
seems to be worst from the point of view of AoI. Indeed, it makes no sense to store an
accepted message if we are only interested in the age of information.
We should process it as soon as possible, perhaps even by preempting
the message that is currently being processed.
It is therefore intuitive that the buffer should have capacity of at most 2,
including the packet currently being processed.
It seems that adding additional buffer space beyond $2$ works against us.
See  \cite{Altman18} regarding this point.

Let us define a system, that we call $\PO_2$.
The buffer has size 2. An arriving message, say message 1, 
to an empty buffer starts being
processed immediately. If a second message, say $2$, arrives while $1$ is being
processed it is stored. If message $3$ arrives while $1$ is still being processed and
$2$ stored, it pushes $2$ out and replaces it.
If no message arrives for a while, then $1$ finishes and $3$ starts being processed,
leaving one unit available to accommodate the next arriving message, if any.
The point is that while a message is being processed, it is never disturbed by
an arriving message. An arriving message will only disturb the stored message,
if any. Could, then, adding an extra unit buffer improve the system from the point
of view of AoI or NAoI? The answer seems to be no. For evidence via simulations,
look, for example at the D/M case (deterministic periodic arrivals, i.i.d.\ exponential 
service times). Figure \ref{fig:dm1} compares the three policies from the
point of view of mean NAoI in steady state as a function of the arrival rate $\lambda$.

\begin{figure}[h]
\centering
\includegraphics[width=8cm]{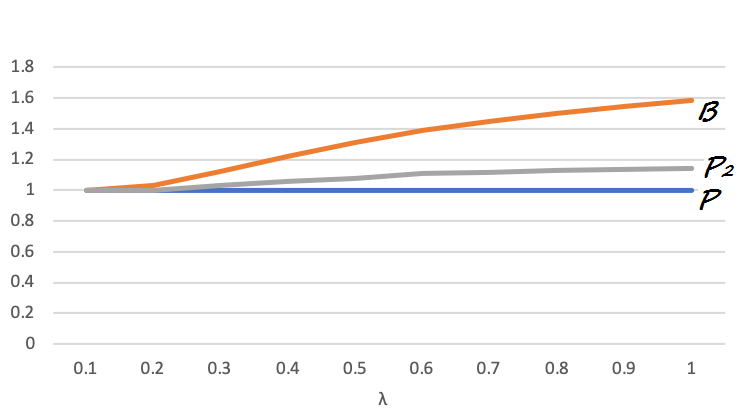}
\caption{ \it
The mean NAoI as a function of $\lambda$ for D/M systems
with deterministic interarrival times
and i.i.d.\ exponential service times (with mean $1$)
in three cases.
}
\label{fig:dm1}
\end{figure}


We next take a look at the most commonly studied system,
an infinite buffer FIFO system. 
But it performs even more poorly.
Consider the $M/M$ case.
Here there is an explicit formula:
\begin{equation}
\text{Mean AoI for M/M/1/$\infty$/FIFO } 
=\frac{1}{\lambda} + \frac{1}{\mu}  +  \frac{\lambda^2}{\mu^2}\cdot
\frac{1}{\mu-\lambda}, 
\label{mean-fifo}
\end{equation}
see \cite[eq.\ (17)]{Yates12}
and compare it with the mean AoI for $\PO$ and $\BL$,
formulas as in Table \ref{tab:means}:
\[
\wE \alpha_\PO(0) = \frac{1}{\lambda} + \frac{1}{\mu} , \quad
\wE \alpha_\BL(0) = \frac{1}{\lambda} + \frac{1}{\mu} + \frac{\lambda}{\mu}\cdot
\frac{1}{\lambda+\mu}.
\]
Clearly, $\wE \alpha_\PO(0)$ is the smallest of all. The FIFO mean
is larger  than $\wE \alpha_\BL(0)$ when $\lambda > \sqrt{2}-1$.
But even when $\lambda < \sqrt{2}-1$, the mean NAoI under FIFO
is only $0.94\%$ better than $\wE \alpha_\BL(0)$. 
See Figure \ref{fig:fifo}; here, the curve for $\PO_2$ has been obtained
by stochastic simulation.
All curves tend to $\infty$ as $\lambda\to 0$ because of the fact
that we plot AoI and not NAoI, and AoI also measures the
time until the previous arrival.
Also, the FIFO curve tends to $\infty$ as $\lambda \to \mu=1$,
and that is because the FIFO system becomes unstable.
\begin{figure}[h]
\centering
\includegraphics[width=8cm]{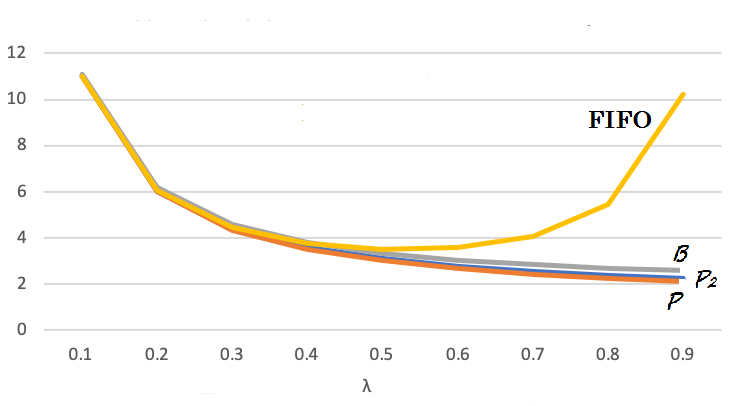}
\includegraphics[width=8cm]{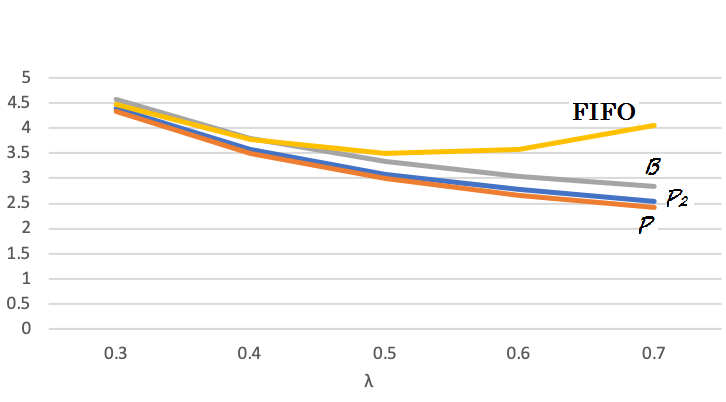}
\caption{\it
The mean AoI as a function of $\lambda$ for M/M systems
with i.i.d.\ exponential interarrival times, i.i.d.\ exponential service times 
(with mean $1$)
in four cases: the bufferless $\PO$ and $\BL$ systems, the $\PO_2$ system,
and the infinite capacity system under FIFO. The right figure is a detail of the left.
Note that $\BL$ is slightly worse than FIFO for small $\lambda$.
However, $\PO_2$ is better than FIFO and $\PO$ is better than $\PO_2$
for all $\lambda$.}
\label{fig:fifo}
\end{figure}

The following observations provide additional evidence regarding the claim
that small buffer systems perform at least as well as the well-studied 
infinite buffer LIFO and FIFO systems.
Consider a sequence of arrival times and processing times fed to 
four systems, 
preemptive infinite 
LIFO buffer (\pLIFO), infinite buffer FIFO, 
the $\PO$ system, and the $\PO_2$ system.
Let $\alpha_\pLIFO$, $\beta_\LIFO$, etc.,
be the AoIs and NAoIs for the four systems.

\paragraph{Observation 1.}
For all times $t$,
$\alpha_\pLIFO(t) = \alpha_{\PO}(t)$
and 
$\beta_\pLIFO(t) = \beta_{\PO}(t)$.\\
To see this, recall that we use the same arrival times and same processing
times for both $\pLIFO$ and $\PO$ systems.
Consider a trajectory of the $\pLIFO$ system. 
We will show how to construct the trajectory of the $\PO$ system deterministically
from that of the $\pLIFO$ system.
Consider  the arrival of a message at an empty $\pLIFO$ system,
call this message 1, letting $2,3,\ldots$ be the indices of subsequent messages.
Observe the $\pLIFO$ system until the end of the busy period started with message $1$.
Necessarily, this busy period also ends with $1$.
Let $j_1$ be the index of the first message within this busy period that will not be 
disturbed by any arriving message: message $j_1$ is processed without
interruption. Similarly, let $j_2>j_1$ be the next uninterrupted message,
and so on.
We can construct the trajectory of $\PO$ by observing that the messages
$1,\ldots,j_1-1$ are unsuccessful and $j_1$ is the first successful message.
Similarly, $j_1+1, \ldots, j_2-1$ are unsuccessful and $j_2$ is successful.
The process $A^*_t$ changes only at the departure times
of $j_1, j_2, \ldots$. It is thus the same for both $\pLIFO$ and $\PO$.
Hence $\alpha_{\pLIFO}(t) = \alpha_\PO(t)$ for all $t$.
Since $A_t$, the last arrival before $t$, is also the same for both systems
(arrivals are coupled), we also have $\beta_\pLIFO(t) = \beta_{\PO}(t)$.

\paragraph{Observation 2.}
For times $t$ that are successful departures from the $\PO_2$ system,
$\alpha_{\FIFO}(t) \ge \alpha_{\PO_2}(t)$ and $\beta_{\FIFO}(t) \ge \beta_{\PO_2}(t)$.
\\
Assume both systems empty at time 0 with the first message
indexed 1 arriving at time $T_1 \geq 0$.
Under system $x\in\{\FIFO, \PO_2\}$, let
$T_{x,k}' $ be  departure time of  message $k$,
$\alpha_x(t)$ be the AoI at time $t$, and 
$W_x(t)$ be the work-to-be-done at time $t$ 
(including the service time of an arrival at time $t$).
For the $\mathcal{P}_2$  system, let 
$\chi_k$ indicate whether the
$k^{\rm th}$ message is successfully served and
let $(k)$ be the index of the $k^{\rm th}$ successfully served message.
Thus, $\forall k$, $\chi_{(k)}\equiv 1$.
We argue inductively that
\begin{equation}\label{W-induction}
\forall k\geq 1, ~~
W_{\rm FIFO}(T_{(k)}) ~ \geq ~
W_{\mathcal{P}_2 }(T_{(k)}).
\end{equation}
Clearly  (\ref{W-induction}) is true with equality ($=\sigma_1$) 
at $k=1$.
Assume  (\ref{W-induction}) for arbitrary $k\geq 1$.
This leads to the following inequality:
\begin{eqnarray*}
\lefteqn{ W_{\rm FIFO}(T_{(k+1)})} & &  \\
& = & \max\left\{W_{\rm FIFO}(T_{(k)}) +
\sum_{i=(k)+1}^{(k+1)}\sigma_i -(T_{(k+1)}-T_{(k)}), \,
\max_{(k)+1\leq \ell \leq (k+1)}
\sum_{i=\ell}^{(k+1)}\sigma_i -(T_{(k+1)}-T_{\ell})\right\}\\
& \ge &  \max\left\{ W_{\mathcal{P}_2}(T_{(k)}) +
\sum_{i=(k)+1}^{(k+1)}\sigma_i\chi_i -(T_{(k+1)}-T_{(k)}), \,
\max_{(k)+1\leq \ell \leq (k+1)}
\sum_{i=\ell}^{(k+1)}\sigma_i\chi_i -(T_{(k+1)}-T_{\ell})\right\}
\\
& = &  \max\left\{ W_{\mathcal{P}_2}(T_{(k)}) +
\sigma_{(k+1)} -(T_{(k+1)}-T_{(k)}), \,\,
\sigma_{(k+1)}\right\}
\\
&= &  W_{\mathcal{P}_2}(T_{(k+1)}) 
\end{eqnarray*}
Thus, for all $k$, the departure time of $(k)$ under FIFO,
\begin{eqnarray*}
T_{\rm FIFO,(k)}'=T_{(k)} + W_{\rm FIFO}(T_{(k)})  
&\geq &  
T_{(k)} + W_{\mathcal{P}_2}(T_{(k)})  
=T_{\mathcal{P}_2,(k)}'.
\end{eqnarray*}
So, at time  $T_{\mathcal{P}_2,(k)}'$, the index of
the most recent completely served message under
FIFO is $k' \leq (k)$.
Therefore,
\[
\alpha_{\mathcal{P}_2}(T_{\mathcal{P}_2,(k)}') =  T_{\mathcal{P}_2,(k)}' - T_{(k)}
\le T_{\mathcal{P}_2,(k)}' - T_{k'}
 = \alpha_{\rm FIFO} (T_{\mathcal{P}_2,(k)}') 
\]
These arguments also hold for  NAoI $\beta$ because arrival
times are coupled.

\paragraph{}
So it is unclear and rather puzzling why infinite capacity systems 
have been considered. As mentioned in Section \ref{bb}, \cite{Shroff17}
showed that, among all infinite capacity systems, and under specific
distributional assumptions, \pLIFO  ~is best.
But \pLIFO ~with preemption has the same  AoI and NAoI as $\PO$.
So all the results obtained in this paper for $\PO$ also hold for
\pLIFO.

For the M/M/1/$\infty$-FIFO system, \cite{Yates12} observes that,
when the service rate $\mu$ is  fixed
the mean AoI is minimized at $\lambda \approx 0.53\mu$.
This is trivial: just minimize the expression \eqref{mean-fifo} over $0<\lambda<\mu$.
To accomplish this may  require that the arrival rate $\lambda$ can be controlled.
It is unclear how this control avoids dropping arriving (freshest) messages, or not
generating them in the first place.  
Obviously, given $\lambda$, one can generally reduce
the mean AoI by increasing $\mu$.
The assumption of a  FIFO queueing discipline has been justified by 
its existing deployment in many practical scenarios  (e.g., message
transmission buffers of sensors).
But practical scenarios also involve {\em finite} message buffers,
and an arriving (freshest)  message to a full buffer is dropped
(unless pushout is available, in which case the queue could be operated
as a bufferless system with pushout).
This will be particularly problematic under heavy traffic.
Generally, the AoI concept is not very interesting under light traffic.

\paragraph{Alternative definitions of age of information.}
Alternative definitions of age of information are possible and may be desirable.
For example, a measure of freshness of information may
involve message streams where the most recent message does
{\it not} obsolete all previous ones. More specifically, assume that,
upon arrival of a new message (with normalized ``importance" 1),
the importance of all prior messages  can be diminished
by a positive factor $\xi<1$, and the objective could be
to minimize the sum of the importance of all transmitted messages.
That is, it may be desirable at the receiver to accurately interpolate
between the freshest messages.
This case may require a large message buffer under LIFO.

\paragraph{Open problems.}$~$
\begin{enumerate}
\item
Compute the distributions and/or expectations of AoI and NAoI
under the $\PO\BL(\ell)$ and $\BL\PO(\ell)$ policies, 
especially under renewal assumptions.
Choose the $\ell$ that minimizes a given performance measure,
e.g., $\wE \beta(0)$ or $\wP(\beta(0)>x)$ as a function of the
interarrival and processing time distributions.

\item
Formulate and solve a dynamic optimization problem.
That is, decide the policy that accepts/rejects incoming messages
and also decides which of them will be successful or not.
Even under renewal assumptions, this is not an easy problem.
Life can possibly be made easier under specific distributional assumptions,
e.g., in the good old M/M case.

\item
Analyze the $\PO_2$ system. That is, compute distributions and/or
expectations for AoI and NAoI.

\item
Conjecture: In steady state,
$\beta_{\PO_2}(0)$ is stochastically smaller than $\beta_{\FIFO}(0)$.
Evidence for this is Observation 2 above.

\item
Conjecture: In steady state,
$\beta_{\PO_2}(0)$ is stochastically smaller than $\beta_{\rm npLIFO}(0)$, where the 
latter is LIFO with {\em non}-preemptive service policy.

\item
Take into account the technological constraints and see if alternative
measures of the age of information can justify large buffers.

\item
Better explain how these measures help real-time systems in real-life situations.
\end{enumerate}

\section*{Acknowledgments}
We thank Kostya Borovkov for reading an early draft of the paper.
We also thank the three anonymous reviewers who read the paper carefully,
identified all typos, some them in the formulas themselves, and critically questioned it.
Their detailed comments helped us improve the paper.

\small
\appendix
\section{List of symbols}
\begin{tabular}{ll}
$\delta_x$ & delta measure at the point $x$ \\
$\overline X$ & a random variable with density $\P(X>x)/\E X$ \\
$T_n$ & arrival time of a message\\
$\chi_n$ & accept/reject index \\
$\psi_n$ & success/failure index \\
$\mathcal Z_n$ & informally, the event that the server is idle just before $T_n$ \\
$\widetilde \P$ & informally, probability measure governing the stationary system\\
$\P$ & Palm probability of $\widetilde \P$ with respect to the arrival process\\
$\P^*$ & Palm probability of $\widetilde \P$ with respect to reading intervals
beginnings $=\P(\cdot|\mathcal Z_0)$\\
$\arr$ &= $\sum_n \delta_{T_n}$, arrival process as a point process\\
$\arr(t)$ & $=\mathfrak a([0,t))$\\
$U(t)$ & $=\E \arr(t)$\\
$Z(t)$ &  $= \E \sum_{i=0}^{\arr(t)-1} T_i$\\
$W_u(t)$ &  $= \E e^{-u T_{\arr(t)}}$\\
$V_u(t)$ & $= \E \sum_{i=0}^{\arr(t)-1}  e^{-u T_i}$\\
$Q_u(t)$ & $= \E  \big\{ (T_{\arr(t)}-t)\, e^{-u T_{\arr(t)-1}} \big\}$\\
$\sigma_n$ & processing time of a message\\
$T_n'$ & departure time of a message either due to successful reading or not\\
$A_t$ & last arrival epoch before $t$ \\
$S_t$ & last arrival epoch before $t$ of a successful message \\
$D_t$ & last departure epoch before $t$ of a successful message \\
$A^*_t$ & $=S_{D_t}$ \\
$\Delta f(t)$ & $=f(t+)-f(t-)$ \\
$\tau_n$ & $=T_{n+1}-T_n$ \\
$B_k$ & beginning of a reading interval \\
$\mathbf R_k$ & duration of a reading interval \\
$B_k'$ & $=B_k+\mathbf R_k$ \\
$\mathbf C_k$ & $=B_{k+1}-_k$, cycle length \\
$\lambda$ & arrival rate $=1/\E \tau$ $=\sum_n \widetilde \P(0<T_n<1)$
\end{tabular}

\bibliographystyle{plain}

\vspace*{1cm}

\begin{center}
\begin{tabular}{ccc}
\begin{minipage}[t]{0.3\textwidth}
\small
{\sc George Kesidis}\\ 
Computer Science Department, 
The Pennsylvania State University, University Park, PA, 16802, USA,
\href{mailto:gik2@psu.edu}{gik2@psu.edu}
\end{minipage}
&
\begin{minipage}[t]{0.3\textwidth}
\small
{\sc Takis Konstantopoulos}\\
Department of Mathematical Sciences , 
The University of Liverpool, Liverpool  L69 7ZL, UK;
\href{mailto:takiskonst@gmail.com}{takiskonst@gmail.com}
\end{minipage}
&
\begin{minipage}[t]{0.33\textwidth}
\small
{\sc Michael A.\ Zazanis}\\ 
Department of Statistics, Athens University of Economics and Business, 76 Patission St., Athens 104 34, Greece;
\href{mailto:zazanis@aueb.gr}{zazanis@aueb.gr}
\end{minipage}
\end{tabular}
\end{center}

\end{document}